\documentclass[twocolumn,prl, superscriptaddress, showpacs]{revtex4}
\newcommand{\pagenumbaa}{1}
\usepackage{graphicx}
\usepackage{amssymb}
\usepackage{amstext}
\usepackage{epsfig}

\usepackage{latexsym}
\usepackage{amsmath}
\usepackage{bm} 
\usepackage{bbm}

\usepackage{color}

\usepackage{amsthm}
\usepackage{amsmath}
\usepackage{amsfonts}
\usepackage{amssymb,amstext}

\theoremstyle{plain}
\newtheorem{theorem}{Theorem}
\newtheorem{lemma}[theorem]{Lemma}

\theoremstyle{definition}

\newcommand{\beq}{\begin{equation}}
\newcommand{\eeq}{\end{equation}}

\newcommand{\beqa}{\begin{eqnarray}}
\newcommand{\eeqa}{\end{eqnarray}}

\newcommand{\bal}{\begin{align}}
\newcommand{\eal}{\end{align}}

\newcommand{\bsp}{\begin{equation}\begin{split}}
\newcommand{\esp}{\end{split}\end{equation}}

\newcommand{\bit}{\begin{itemize}}
\newcommand{\eit}{\end{itemize}}

\newcommand{\ben}{\begin{enumerate}}
\newcommand{\een}{\end{enumerate}}

\newcommand{\nn}{\nonumber}
\newcommand{\SPAN}{\text{span}}
\newcommand{\schweif}[1]{\mathcal{#1}}

\newcommand{\Eig}[2]{\text{Eig}(#1,#2)}
\renewcommand{\sp}[2]{\langle #1,#2 \rangle}
\newcommand{\AR}{\mathbb{R}}
\newcommand{\EW}[1]{\mathbb{E} \left[  #1 \right] }
\newcommand{\avEW}[1]{\hat{\mathbb{E}}_{I} \left[  #1 \right] }

\newcommand*{\ket}[1]{| #1 \rangle}
\newcommand*{\bra}[1]{\langle #1 |}

\newcommand{\HR}{\mathcal{H}}

\newcommand{\deltafunc}[1]{\delta_{\mathbf{#1}}}
\newcommand{\distH}[2]{\text{dist}_{\HR}(\mathbf{#1}, \mathbf{#2})  }
\newcommand{\distHJK}[2]{\text{dist}_{\HR}^{(J,K)}(\mathbf{#1}, \mathbf{#2})  }

\newcommand{\supp}[1]{\mathrm{supp}\{ #1 \} }

\newcommand{\dirac}[1]{\delta_{\mathbf{#1}} }

\newcommand{\Image}[1]{\text{Im}(#1)}

\begin{document}


\title{Localization of Toric Code Defects}


\author{Cyril Stark}

\affiliation
{Theoretische Physik, ETH Zurich, CH-8093 Zurich, Switzerland}

\author{Lode Pollet}

\affiliation
{Theoretische Physik, ETH Zurich, CH-8093 Zurich, Switzerland}

\author{Ata\c{c} Imamo\u{g}lu}

\affiliation
{Institute for Quantum Electronics, ETH Zurich, CH-8093 Zurich, Switzerland}

\author{Renato Renner}

\affiliation
{Theoretische Physik, ETH Zurich, CH-8093 Zurich, Switzerland}


\begin{abstract}

We explore the possibility of passive error correction in the toric code model. We first show that even coherent dynamics, stemming from spin interactions or the coupling to an external magnetic field,  lead to logical errors. We then argue that Anderson localization of the defects, arising from unavoidable fluctuations of the coupling constants, provides a remedy. This protection is demonstrated using general analytical arguments that are complemented with numerical results which demonstrate that self-correcting memory can in principle be achieved in the limit of a nonzero density of identical defects.

\end{abstract}


\maketitle

\setcounter{page}{\pagenumbaa}
\thispagestyle{plain}


Classical information can be reliably  stored by encoding it in
long-living metastable states of a many-particle system (e.g., the
magnetic surface of a disk). The code states are typically states
carrying different values of an order parameter, separated by an
energy barrier that grows (in the optimal case linearly) with the
system size $N$~\cite{Alickietal07,Pastawskietal10}. Consequently,
the probability that an error occurs decreases exponentially in $N$,
rendering a high robustness of (even very small) memory devices.
Crucially, the stability of such devices is due to the intrinsic
(local) interactions between the particles, and no active error
correction is required during storage. They are thus
often called 'passive' or 'self-correcting' memories.

The situation appears far more challenging when  it comes to the
storage of quantum information. While quantum information can in
principle~\footnote{The positive results on error correction are
almost exclusively based on models that assume limited local (but
possibly correlated) errors.}  be protected against disturbances
using active error-correction~\cite{Shor95,Steane96a,Steane96b}, to
date no realistic many-body system is known to passively (i.e., by
virtue of its natural dynamics) preserve quantum information over
macroscopic timescales. Nevertheless, self-correcting properties may
be obtained by a clever design of the Hamiltonian of the system.  In a
pioneering paper~\cite{Kitaev2003,DKLP02}, Kitaev proposed the toric
code as a topologically protected quantum memory, where
information is stored in degenerate (and locally indistinguishable)
ground states. His proposal prompted an intensive study of the use
and limitations of such systems for information storage (see, e.g., \cite{Bacon08} for an overview). 

In the toric code excitations out of the ground state space can be
removed by active error correction. However, if the corresponding defects have traced out un-contractible loops when
annihilated, this causes an error. In particular, it has been
demonstrated that Kitaev's 2-dimensional toric code cannot provide
robustness against the destructive influence of a thermal
environment or against Hamiltonian perturbations which lead to a
random walk of the thermally excited defects~\cite{AliHor06,CastCham07,Alickietal07,Trebstetal07,NussOrt08,Kay09,Iblisetal09,Vidaletal09,Tupitsynetal10}.
Furthermore, these impossibility results have been generalized to a
wide class of 2-dimensional lattice
systems~\cite{KayCol08,BravTerh09,Pastawskietal10}. On the positive
side, a number of variants of topologically protected systems have been
proposed where the relevant energy barrier grows with the system
size~\cite{Hamma2009, DKLP02, Bacon08, Chesietal10, FabioChesiLoss10}.

In this Letter, we analyze passive error correction in the presence of a large class of spin interactions or the coupling to an external magnetic field, as well as unavoidable or engineered fluctuations in the toric code coupling constants. We demonstrate that coherent defect propagation~\footnote{This must not be confused with propagation caused by
thermal hopping.} would have lead to logical errors, had it not been for the Anderson localization induced by the fluctuations in the coupling constants. In proposed realizations of the toric code as an effective model, e.g., from an
underlying Kitaev's honeycomb model~\cite{Kitaev2006}, perturbations at the `physical level'
(e.g. dipolar interactions in the honeycomb model) induce complicated
perturbations at the effective level of the toric code, which still can be
approximated within the class of interactions considered in this Letter. The conclusion that the stability of the toric code can be improved by randomness was also reached in~\cite{Castelnovo2010} where it was shown that the stability of the topological entanglement entropy is enhanced by the presence of random magnetic fields.

\emph{Error Model---}Consider a square lattice $\Gamma$, with spins sitting on its edges, embedded in an arbitrary two-dimensional manifold. Assume the dynamics is described by the toric code Hamiltonian
\beq\label{Def.H_TC}
	H_{TC} = -J_{m} \sum_{p} \prod_{j \in p} Z_{j} - J_{e} \sum_{s} 	\prod_{j \in s} X_{j} 
\eeq 
($J_{m} > 0$, $J_{e} > 0$, see \cite{Kitaev2003}) where the sums run over all plaquettes $p$ and over all stars $s$, respectively. Eigenvectors $\ket{\psi}$ of $H_{TC}$ with $\prod_{j\in p}Z_{j} \ket{\psi} = - \ket{\psi}$ are said to have a 'magnetic defect' at plaquette $p$. Analogously, an 'electric defect' at star $s$ is detected by $\prod_{j\in s}X_{j}$. Defects of same type are bosons among themselves. Defects of opposite type have mutual exotic braiding phases.  This setup serves as a model for topological quantum memories.  However, it will never be possible to realize the $H_{TC}$-dynamics perfectly in any experiment. In particular, time-independent perturbations will always be present in any experimental realization of the toric code. Two physically meaningful examples of such perturbations are the dipole-dipole interaction
\beq
	H_{\mathrm{d}} = \eta \sum_{i,j, i\neq j} \frac{\vec{\sigma}_{i} \cdot \vec{\sigma}_{j}}{\| \vec{r_{ij}} \|^3} - 3 \frac{  \left( \vec{\sigma}_{i} \cdot \vec{r}_{ij} \right) \left( \vec{\sigma}_{j} \cdot \vec{r}_{ij} \right)  }{\| \vec{r}_{ij} \|^5}
\eeq
($\eta$ collects constant factors) with $\vec{r}_{ij} :=  \vec{r}_{i} -\vec{r}_{j} $, and the effect of an external homogeneous magnetic field
\beq\label{Spin.Hamiltonian.for.external.magnetic.field}
	H_{\mathrm{magnetic}} = \tilde{\eta} \sum_{j} \vec{\sigma}_{j} \cdot \vec{B},
\eeq
where $\vec{\sigma}_{j} = (X_{j},Y_{j},Z_{j})$ denotes the vector of Pauli matrices at spin/edge $j$ and $\tilde{\eta}$ collects constant factors. In this Letter we consider the effect of translation invariant perturbations of the form 
\beq\label{Def.H.I}
	H_{I} = \sum_{\substack{i_{1},...,i_{m},\\  \alpha_{1}, ..., \alpha_{m} \in \{x,y,z\}}}     \xi(i_{1},\alpha_{1},...,i_{m},\alpha_{m}) \vec{\sigma}_{i_{1},\alpha_{1}} \cdots \vec{\sigma}_{i_{m},\alpha_{m}}.
\eeq
Here we sum over all collections of $m$ distinct spins/edges $i_{1}, ..., i_{m}$ and all possible choices of corresponding Pauli operators $\alpha_{1}, ..., \alpha_{m} \in \{ x,y,z \}$. The function $\xi$ has a spatial cutoff, such that $\xi(i_{1}, ..., i_{m}) = 0$ whenever $\max\{ \| i_{k} - i_{l} \|  :  k,l \in \{ 1,...,m \}  \} > R$  ($R < \infty$). The perturbations of the form~\eqref{Def.H.I} can be seen as simple models for interactions of the dipolar type, but also include the perturbation $H_{\mathrm{magnetic}}$ in Eq.~\eqref{Spin.Hamiltonian.for.external.magnetic.field}. The perturbed Hamiltonian reads $H = H_{TC} + H_{I}$. It has been shown~\cite{BravyiHasitings2010} that perturbations of the form~\eqref{Def.H.I} lead to a splitting of the ground state degeneracy that is exponentially small in the system size. 

In lowest order degenerate perturbation theory one approximates the Hamiltonian $H$ by the operator $\sum_{n_{e}, n_{m}} P_{n_{e}n_{m}} H P_{n_{e}n_{m}}$, where $P_{n_{e}, n_{m}}$ is the projector onto the unperturbed eigenspace carrying $n_{e}$ pairs of electric defects and $n_{m}$ pairs of magnetic defects. The pairwise orthogonality of the projectors $P_{n_{e}n_{m}}$ allows for separate treatments of the dynamics that is induced by each of the operators $P_{n_{e}n_{m}} H P_{n_{e}n_{m}}$. Consider the simple special case $n_{e} = 1$, $n_{m} = 0$. To figure out the nature of the dynamics described by $P_{1,0} H_{I} P_{1,0}$ we compute the matrix elements of $P_{1,0} H_{I} P_{1,0}$ with respect to the basis built up by the toric code eigenstates living in the image of $P_{1, 0}$. Terms that either contain an $X_{j}$ or a $Y_{j} = i X_{j}Z_{j}$ give no contribution because they create new magnetic defects (or annihilate existing ones) and consequently map states in the image of $P_{1, 0}$ to its orthogonal complement. For the computation of the matrix elements we are left with the following expression:
\begin{multline}\label{ourdynamics}
	\bra{g_{i}} \prod_{s \in l_{1}} Z_{s}  C \prod_{t \in l_{2}} Z_{t} \ket{g_{j}}
	=	\sum_{i_{1},...,i_{m}: i_{q} \neq i_{t}} \xi(i_{1},z,...,i_{m},z)   \\
		\times \bra{g_{i}} Z_{s_{1}} \cdots Z_{s_{u}} \,    Z_{i_{1}} \cdots Z_{i_{m}} \,  Z_{t_{1}} \cdots Z_{t_{v}}    \ket{g_{j}},
\end{multline}
where
\beq
	C := \sum_{i_{1},...,i_{m}: i_{q} \neq i_{t}} \xi(i_{1},z,...,i_{m},z) \, Z_{i_{1}} \cdots Z_{i_{m}}.
\eeq
and $l_{1} = \{ s_{1}, ..., s_{u} \}$, $l_{2} = \{ t_{1}, ..., t_{v} \}$ are paths on the lattice, and where the vectors $\{ \ket{g_{j}} \}_{j}$ form a basis of the ground state space (the vectors of the form $\prod_{s \in l_{2}} Z_{s} \ket{g_{j}}$ build up a basis in the image of $P_{1,0}$; cf.~\cite{Kitaev2003}). The matrix element within the sum is nonzero if and only if the paths $l_{1}$ and $l_{2}$ differ by the movement of the electric defects that is determined by the action of $Z_{i_{1}} \cdots Z_{i_{m}}$. Note that $P_{1,0} H_{I} P_{1,0}$ does not allow the particles to sit on top of each other (this corresponds to the annihilation of defects). It thus follows that $P_{1,0} H_{I} P_{1,0}$ describes a finite range hopping term with an effective hardcore repulsion. The evolution of the pair of electric defects determined by the spin-Hamiltonian $P_{1,0} H_{I} P_{1,0}$ can equivalently be described by a Hamiltonian $T$ for the evolution of the defects on the lattice $\Gamma$. These observations generalize---up to exotic braiding phases between defects of different type---to arbitrarily many defects~\footnote{Then, only defects of the same type are hardcore repulsive. Our findings are invariant under charge conjugation.}. In case of a homogeneous external magnetic field along the $z$-direction~\footnote{This model may alternatively be analyzed within the low-energy Ising model description~\cite{HammaLidar2008} of the toric code.} we have an intuitive understanding for the nature of the induced dynamics (nearest neighbor hopping of the electric defects with hardcore repulsion): the wave function of a nearest neighbor pair of electric defects will spread arbitrarily during time evolution (cf. the inset of Fig.~\ref{fig:All_in_one}). Note that this leads to the failure of active error correction (fusion of nearest neighbors) at the read out, and the memory will thus become unstable. Next, we will show that this problem will be present for all perturbations of the form~\eqref{Def.H.I}. To apply methods from spectral theory we assume $\Gamma$ to be $\mathbb{Z}^2$.

\emph{Propagation to Infinity---} The goal of this section is to show that there exist initial 2-defect wave functions (e.g., $(n_{e},n_{m}) = (1,0)$) with the property that the two defects travel arbitrarily far away from each other. This eventually leads to an error in the logical qubits. For this we need to determine the Hamiltonians that govern their relative motion. We follow the approach described in~\cite{Graf1997} and~\cite{Albeverio2006} and find that the relative dynamics is generated by a family of hopping Hamiltonians $T(k)_{0} + T(k)_{I}$ acting on the \emph{relative} Hilbert space $l^{2}(\mathbb{Z}^2)$ that is parameterized by the quasi-momentum $k$ ($k \in [0,2\pi)^2$). The operator $T(k)_{0}$ describes translationally invariant, finite-ranged hopping on the relative configuration space $\mathbb{Z}^2$ (i.e., the set of vectors that connect the two defects). The interaction $T(k)_{I}$ on the other hand dictates inhomogeneous hopping and is only supported within a finite neighborhood of the origin in the configuration space of relative motion. In~\cite{EPAPS} we prove that there exist initial states such that the two particles are not found within finite relative distance as time approaches infinity:
\begin{theorem}
	For any quasi-momentum $k$ there exist initial states $\psi_{0}^{(k)} \in l^2(\mathbb{Z}^2)$ for the relative dynamics such that $\lim_{t \rightarrow \infty} \sum_{x \in \Lambda}  \left|   e^{-i  (T(k)_{0} + T(k)_{I} )  t}\psi_{0}^{(k)}(x)   \right|^2 = 0$ for any finite subset $\Lambda$ of the relative configuration space $\mathbb{Z}^2$.
\end{theorem}

\emph{Effective Random Potential---}Define
\beq\label{Def.H.r}
	H_{r} := - \sum_{p} J_{m}(p) \prod_{j \in p} Z_{j} -  \sum_{s} J_{e}(s) \prod_{j \in s} X_{j}, 
\eeq	
where $J_{m}(p)$ and $J_{e}(s)$ are positive iid random variables described by a bounded and compactly supported probability density. In lowest order degenerate perturbation theory the total Hamiltonian now reads $\sum_{n_{e}, n_{m}} P_{n_{e}n_{m}} \left( E_{n_{e}n_{m}} + \lambda  H_{r} + H_{I} \right) P_{n_{e}n_{m}}$. The terms $P_{n_{e}n_{m}} H_{I} P_{n_{e}n_{m}}$ have been discussed before. We are left with the computation of the matrix elements of the operators $P_{n_{e}n_{m}} H_{r} P_{n_{e}n_{m}}$ with respect to the eigenbasis of the positions of the magnetic and electric defects. Since the vectors in this basis are automatically eigenvectors of $P_{n_{e}n_{m}} H_{r} P_{n_{e}n_{m}}$ we conclude that $P_{n_{e}n_{m}} H_{r} P_{n_{e}n_{m}}$ acts as a multiplication operator on the elements of the position-eigenbasis. The operators $P_{n_{e}n_{m}} H_{r} P_{n_{e}n_{m}}$ thus play the role of potential terms. The potential felt by electric defects moving on $\Gamma$ is specified by iid random variables $V_{s} = 2 J_{e}(s)$ (and analogously for the magnetic defects).

\emph{Localization---}From the work emanating from Anderson's discovery in 1958 we expect that the iid potential that is caused by the randomization of the coupling constants leads to Anderson localization of the two defects. In \cite{AizenmanWarzel2009} Aizenman and Warzel proved dynamical localization of interacting $n$-body systems ($n < \infty$) on $\mathbb{Z}^d$ under the assumption that the interactions are described by interaction potentials with finite range: 
\beq\label{Aizenman.Warzel.Hamiltonian}
	H^{(n)} = \sum_{j=1}^n \left[ -\Delta_{j} + \lambda V(x_{j}) \right] + \mathcal{U}(\mathbf{x};\mathbf{\alpha})
\eeq
($\mathbf{x} = (x_{1},x_{2},...,x_{n}) \in \mathbb{Z}^{dn}$). Here, $\Delta_{j}$, $V(x)$, and $\mathcal{U}(\mathbf{x};\mathbf{\alpha})$ denote the discrete Laplacian, iid random one-particle potential, and interaction potential, respectively (cf.~\cite{AizenmanWarzel2009}). They prove dynamical localization with respect to the Hausdorff pseudo distance measure defined by $\text{dist}_{\mathcal{H}}(\mathbf{x}, \mathbf{y}) := \max \left\{    \max_{1\leq i \leq k} \text{dist}(x_{i}, \{ \mathbf{y} \} ) ,  \max_{1\leq i \leq k} \text{dist}(\{ \mathbf{x} \}, y_{i} )    \right\}$. As discussed in~\cite{AizenmanWarzel2009} this result has the drawback that it still allows a particle to hop from one  tight cloud of particles to another. However, such processes appear rather unphysical, and results in~\cite{Chulaevsky2010A, Chulaevsky2010B} suggest that they are unlikely. Note that the Hamiltonian~\eqref{Aizenman.Warzel.Hamiltonian} is not exactly of the form~\eqref{ourdynamics} because in our setting the finite-range interactions between the two defects are given in terms of inhomogeneous hopping matrix elements. To get localization bounds for our 2-particle system the proof of Aizenman and Warzel needs to be adapted. The details are given in~\cite{EPAPS}. We arrive at the following theorem about dynamical localization. 
\begin{theorem}\label{theorem.about.localization}
	Let $H^{(2)}$ be the random Hamiltonian describing the evolution of the pair of electric defects on the lattice $\mathbb{Z}^2$ that corresponds to the spin-Hamiltonian $P_{1,0} (H_{I} +\lambda H_{r}) P_{1,0}$. For each $m \in \mathbb{N}$ (cf.~\eqref{Def.H.I}) there is a $\lambda_{0} \in \AR_{+}$ with the property that for all $\lambda \geq \lambda_{0}$ there exist $A,\xi < \infty$ such that for all $\mathbf{x}, \mathbf{y} \in \mathbb{Z}^{4}$
	\beq\label{dynamical.localization}
		\mathbb{E}\left[ \sup_{t \in \AR} \large|  \langle  \mathbf{y} | e^{-it H^{(2)}} | \mathbf{x}      \rangle   \large|  \right] \leq A \, e^{-\mathrm{dist}_{\mathcal{H}}(\mathbf{x}, \mathbf{y}) / \xi  }.
	\eeq
\end{theorem}

Neglecting the exotic braiding phases acquired by braiding defects of opposite type the theorem generalizes to finitely many defects. Taking into account such phases is still an open problem. However, if only one defect type is dynamic (e.g., in case of~\eqref{Spin.Hamiltonian.for.external.magnetic.field} with $\vec{B}$ along the $z$-direction) the static defects influence the evolution of the dynamic defects in terms of random vector potentials. First steps towards the proof of localization for such systems have been taken in~\cite{Klopp2003} and~\cite{ErdosHasler2011}.

\emph{Numerics: localization of two electric defects---}Consider a $L \times L$ lattice with periodic boundary conditions. The dynamics is described by nearest neighbor hopping together with a background potential that is uniformly distributed on $[0, \Delta]$. The left part of Fig.~\ref{fig:All_in_one} displays the 2-particle dynamics for a typical potential landscape. At time $t=0$ the 2-particle wave function is a position-eigenstate $\ket{\mathbf{x}}$, $\mathbf{x} = (x_{1},x_{2})$ with particles ``1'' and ``2'' being nearest neighbors. For each $\mathbf{y} = (y_{1},y_{2})$ we record the value $\sup_{t \in \{ 1, ..., t_{\max} \}} | \bra{\mathbf{y}} \exp(-i H t) \ket{\mathbf{x}} |$ and plot it as a function of the relative distance $\| y_{1} - y_{2} \|_{1} - 1$  ($\| \cdot \|_{1}$ denotes the 1-norm; $\| y_{1} - y_{2} \|_{1} = 1$ corresponds to nearest neighbor configurations). Each point in the figure thus corresponds to a specific 2-particle configuration $(y_{1},y_{2})$ with $y_{1}$ and $y_{2}$ on the periodic $L \times L$ lattice under consideration ($L=10$). 
\begin{figure}[tbp]
\centering
\includegraphics[width=0.49\textwidth]{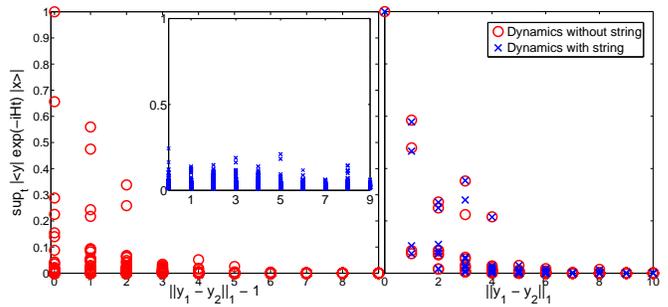}
\caption{(Color online) Left: Anderson localization of the two particles ($\Delta = 50$, $L = 10$, $t_{\max} = 60$). Inset: Dynamics of the two particles in absence of a random background potential. Right: Comparison between the 1-particle dynamics with and without the exotic braiding phases ($\Delta = 50$, $L = 40$, $t_{\max} = 100$). The figure suggests that localization is not greatly affected by exotic braiding phases.}
\label{fig:All_in_one}
\end{figure}

\emph{Numerics: influence of the exotic braiding phases---}To investigate the effect of braiding phases on localization, we consider a single electric defect in the presence of two static magnetic defects. The pair of magnetic defects are generated by a path operator $\prod_{j \in l} X_{j}$ applied to a vector in the ground state space, with $l$ being a path on the lattice. The nearest neighbor hopping terms that cross $l$ change sign. In the right half of Fig.~\ref{fig:All_in_one} we compare the localization in presence and absence of such a path along a straight  line of length $L/2$. The localization of the electric defect appears unaffected by the presence of the path, and thus of the exotic braiding phases. 

\emph{Numerics: localization at positive densities---} To numerically estimate the stability of the perturbed toric code at positive densities (i.e., infinitely many defects in the infinite system) we use Quantum Monte Carlo worm-type simulations~\cite{Prokofev98} (here in the implementation of Ref.~\cite{Pollet07}) to determine the phase transition from the superfluid phase to the insulating Bose glass phase (cf. Fig.~\ref{fig:DF.BG.transition}), similar as was done in Ref.~\cite{Pollet09}. The property of being insulating suggests that the two defects of a pair do not travel arbitrarily far away from each other. The details are given in~\cite{EPAPS}.
\begin{figure}[tbp]
\centering
\includegraphics[width=0.6\columnwidth, angle=-90]{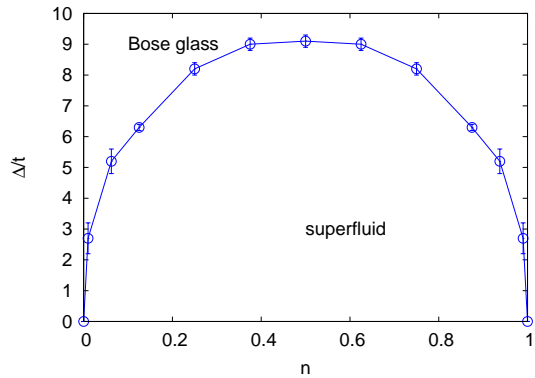}
\caption{(Color online) Plot of the critical disorder $\Delta_{c}$ as a function of the particle density $n$. We assume nearest neighbor hopping (e.g., caused by~\eqref{Spin.Hamiltonian.for.external.magnetic.field}) together with an  iid background potential whose values are uniformly distributed on $[0, 2 \Delta]$. The boundary conditions are periodic and the unit is the hopping $t$.}
\label{fig:DF.BG.transition}
\end{figure}

\emph{Conclusion:} Our results suggest that logical errors caused by interaction-induced propagation of defects can be suppressed by randomness in the toric code Hamiltonian. While the mathematical approach was applicable for a broad class of perturbations and finitely many particles in infinite systems, Quantum Monte Carlo simulations allowed us to draw conclusions for infinite systems with positive defect densities that are exposed to a homogeneous magnetic field along the $z$-axis. Single particle simulations indicate the firmness of our findings against exotic braiding phases.

\emph{Additional Note---}Similar results have been obtained independently in~\cite{WoottonPachos2011} by J. Wootton and J. Pachos.

\emph{Acknowledgment~- }We wish to express our gratitude to Johan {\AA}berg, Charles-Edouard Bardyn, Roger Colbeck, Daniel Egli, J\"urg Fr\"ohlich, Gian Michele Graf and Gang Zhou for many helpful discussions. This research was supported by the Swiss National Science Foundation  through the National Centre of Competence in Research \emph{Quantum Science  and Technology} and under grant PZ00P2-121892. Simulations were performed at the Brutus cluster at ETH Zurich.

\begin{widetext}


\section{Supplementary Material}

\section{Toric Code}\label{Toric.Code}

Consider a square lattice $\Gamma \subset \mathbb{Z}^2$ embedded in an arbitrary two-dimensional manifold with spins sitting on its edges and assume the dynamics being described by the stabilizer Hamiltonian
\beq
	H_{TC} = -J_{m} \sum_{p} \prod_{j \in p} Z_{j} - J_{e} \sum_{s} 	\prod_{j \in s} X_{j} 
\eeq 
($J_{m} > 0$, $J_{e} > 0$) where the sums run over all plaquettes $p$ and over all stars $s$ respectively. This exactly solvable model bears the name Toric Code~\cite{Kitaev2003}. Its groundstate eigenspace 
\beq
	\Eig{H_{TC}}{E_{g}} =  \SPAN\{ \ket{g_{j}} \}_{j = 1,...,2^{2g + h}}
\eeq
is $2^{2g + h}$-times degenerate if the spin-1/2 square lattice is embedded in a manifold with genus $g$ and $h$ holes. The operators
\beqa
	W^{(e)}_{l}	&:=&	\prod_{j \in l} Z_{j}, \nonumber \\
	W^{(m)}_{l^*}	&:=&	\prod_{j \in l^*} X_{j}
\eeqa
(products of $X$- or $Z$-Pauli operators along the paths $l$ on the square lattice $\Gamma$ and $l^*$ on the dual lattice $\Gamma^*$) are sometimes called electric and magnetic path operators. They are convenient to describe excited eigenstates of the Toric Code Hamiltonian $H_{TC}$ because each $H_{TC}$-eigenspace is spanned by vectors that result from the application of some electric and magnetic path operators to vectors in the ground state space:
\beq
	\Eig{H_{TC}}{E_{n_{e},n_{m}}}  = \SPAN \left.\left\{ W^{(e)}_{l_{1}} \cdots W^{(e)}_{l_{n_{e}}} W^{(m)}_{l_{1}^*} \cdots W^{(m)}_{l_{n_{m}}^*} \ket{g_{j}} \right| l_{i}, l_{j}^* \text{ open paths } \right\} .
\eeq
There is a lot of redundancy in this describtion of the eigenspaces because two vectors $W^{(e)}_{l_{1}} \ket{g_{j}}$ and $W^{(e)}_{l_{2}} \ket{g_{j}}$ are equal if (1) the endpoints of $l_{1}$ and $l_{2}$ are equal and (2), if the paths $l_{1}$ and $l_{2}$ are homotopic.  All electric or magnetic path operators along closed paths that can be contracted to one point act as the identity on the ground state space. On the other hand, electric and magnetic path operators that cannot be contracted to one point may cause a non-trivial linear transformation from the ground state space to itself. The energy-observable $H_{TC}$ is everywhere indifferent to the electric and magnetic operator strings except at their endpoints. This leads to the jargon that there are electric charges sitting at the endpoints of electric path operators and magnetic charges sitting at the endpoints of magnetic path operators. The commutation relations Pauli matrices lead to anyonic statistics between charges of different type.

\section{Error correction}\label{Error.correction}

The quantum information that needs to be stored with the help of the toric code is encoded in the degenerate ground state space. If the spin system lives on the torus the ground state space is $2^2$-dimensional and can thus carry 2 qubits. Assume that we have encoded our quantum information into the ground state space at time $t=0$. As time passes the toric code interacts with its environment and the state will acquire support outside of the ground state space due to the interactions with the environment and the imperfect experimental realization of the toric code. Error correction may move the state back to the ground space. A possible algorithm for error correction goes as follows:
\begin{enumerate}
\item Measure $H_{TC}$. This is equivalent to the measurement of all the plaquette observables $-J_{m} \prod_{j \in p} Z_{j}$ and all the star observables $- J_{e}	\prod_{j \in s} X_{j}$. There are magnetic charges associated to the plaquettes $p$ where the measurement of $-J_{m} \prod_{j \in p} Z_{j}$ yields $+J_{m}$ and there are electric charges associated to the stars $s$ wherever the measurement of $- J_{e}	\prod_{j \in s} X_{j}$ yields $+J_{e}$.

\item Pair up all the charges of same type in such a way that the sum of the relative distances between the 2 charges of each pair is minimal. 

\item Fuse the 2 charges associated to each pair by the application of electric path operators $W^{(e)}_{l}$ and magnetic path operators $W^{(m)}_{l^*}$ along shortest paths connecting these 2 charges.
\end{enumerate}

Non-contractable electric and magnetic path operators can form non-trivial maps from the ground state space to itself. Therefore, a possible scenario in case of which the above error correction fails goes as follows: consider an embedding of the toric code into the surface of a torus and assume that sometimes after the initialization of the memory the interactions with the environment causes a pair of magnetic charges. As time passes the imperfect experimental realization of the toric code leads to a movement of the particles in such a way that after a certain time the particles may have moved more than halfway around the torus. Consequently --- after the fusion of the 2 particles in step 3 of error correction --- the ground state vector at $t=0$ and the ground state vector after the error correction differ by the application of a magnetic path operator along a path that cannot be contracted to a point. Such a path operator acts non-trivially on the ground state space and we thus end up with an error in the stored quantum information.

\section{Propagation to Infinity}\label{Propagation.to.Infinity}

The goal of this section is to show that there exist initial 2-defect wave function with the property that these two defects travel arbitrarily far away from each other during time evolution. This will be achieved by showing that the absolutely continuous spectrum of the Hamiltonians generating the relative motion of the two charges is nonempty which --- via the RAGE theorem --- implies the assertion. Let us briefly recall these notions. According to the spectral theorem there exists a spectral measure $\{ P_{\Delta} \}_{\Delta \subseteq \AR}$ for every Hamiltonian $T$ such that 
\beq
	T = \int_{\AR} \lambda dP_{\lambda}.
\eeq
Moreover there exists a decomposition 
\beq
	P_{\Delta} = P^{(pp)}_{\Delta}  +  P^{(ac)}_{\Delta}  +  P^{(sc)}_{\Delta}
\eeq
with the property that the measure
\beq
	\mu_{\phi}(\Delta) := \sp{\phi}{P_{\Delta} \phi}
\eeq
on $\AR$ is pure-point if $\phi \in \Image{P^{(pp)}_{\Delta}}$, absolutely continuous if $\phi \in \Image{P^{(ac)}_{\Delta}}$ and singular continuous if $\phi \in \Image{P^{(sc)}_{\Delta}}$. Define the subspaces $\HR_{pp} = \Image{P^{(pp)}_{\AR}}$, $\HR_{ac} = \Image{P^{(ac)}_{\AR}}$ and $\HR_{sc} = \Image{P^{(sc)}_{\AR}}$. The restrictions $T_{pp} := T|_{\HR_{pp}}$, $T_{ac} := T|_{\HR_{ac}}$ and $T_{sc} := T|_{\HR_{sc}}$ lead to the definitions of the pure-point part $\sigma_{pp}(T) = \sigma(T_{pp})$ of the spectrum, the absolutely continuous  part $\sigma_{ac}(T) = \sigma(T_{ac})$ of the spectrum and singular continuous  part $\sigma_{sc}(T) = \sigma(T_{sc})$ of the spectrum with the property 
\beq
	\sigma(T) = \sigma_{pp}(T) \cup \sigma_{ac}(T) \cup \sigma_{sc}(T).
\eeq
The RAGE theorem implies the following: Let $\psi_{0} \in \HR_{ac}$ and $\Lambda \subset \mathbb{Z}^d$ finite. Then:
\beq\label{RAGE}
	\lim_{t \rightarrow \infty}  \left( \sum_{x \in \Lambda}  \left|   e^{-iTt}\psi_{0}(x)   \right|^2  \right) = 0.
\eeq
In other words --- assuming that $\psi$ is a 1-particle wave function ---  the probability for finding the particle within any finite region $\Lambda \subset \mathbb{Z}^d$ vanishes as time approaches infinity. Hence, if we were able to show that $\HR_{ac} \neq \emptyset$ for the Hamiltonians generating the relative motion between the two charges we would know that there exist initial 2-defect wave functions which will not be found within finite relative distance as time approaches infinity. Consequently, error correction will fail as time approaches infinity. The goal of the remainder of this section is the proof that indeed $\HR_{ac} \neq \emptyset$ for the Hamiltonians generating the relative motion between the two charges.

\subsection{Relative Dynamics}

To show that the absolutely continuous spectrum of the Hamiltonians that generate the relative motion between the two defects is non-empty we first need to determine these Hamiltonians. For systems that live in the continuum $\AR^d$ there is a single Hamiltonian $H_{\mathrm{rel}}$ that describes the relative motion of the two particles with respect to the center-of-mass frame. On the lattice, the situation is a little more delicate. We follow the approach described in~Ref.~\onlinecite{Albeverio2006} that is itself based on Ref.~\onlinecite{Graf1997}: Let $\hat{U}_{s} : l^{2}(\mathbb{Z}^4) \rightarrow l^{2}(\mathbb{Z}^4)$,
\beq
	(\hat{U}_{s} \psi)(n_{1},n_{2}) := \psi(n_{1}-s, n_{2}-s)
\eeq
($s, n_{1}, n_{2} \in \mathbb{Z}^2$) be the representation of the group of simultaneous translations on the 2-particle Hilbert space $l^{2}(\mathbb{Z}^4)$. The discrete Fourier transform $\wedge : L^2(\mathbb{T}^2 \times \mathbb{T}^2) \rightarrow l^{2}(\mathbb{Z}^2 \times \mathbb{Z}^2)$,
\beq
	\hat{\psi}(n_{1},n_{2}) = \frac{1}{(2\pi)^4} \int_{[0,2\pi)^4}d^2k_{1} d^2k_{2} \, \psi(k_{1},k_{2}) e^{-i \mathbf{n} \cdot \mathbf{k}}
\eeq
($\mathbb{T} = [0,2\pi)$) and its inverse  $\vee : l^{2}(\mathbb{Z}^2 \times \mathbb{Z}^2)  \rightarrow L^2(\mathbb{T}^2 \times \mathbb{T}^2)$,
\beq
	\psi(k_{1},k_{2}) = \sum_{\mathbf{n} \in \mathbb{Z}^4}  \hat{\psi}(n_{1},n_{2}) e^{i \mathbf{n} \cdot \mathbf{k}},
\eeq
are unitary. On $L^2(\mathbb{T}^2 \times \mathbb{T}^2)$, the representation of the group of the simultaneous translations of the two particles reads
\beq
	(U_{s}\psi)(k_{1},k_{2}) = \left( \vee \circ \hat{U}_{s} \circ \wedge \psi \right)(k_{1},k_{2}) = e^{i s \cdot (k_{1} + k_{2})} \psi(k_{1},k_{2}) . 
\eeq
We conclude that
\beq
	\left.U_{s}\right|_{L^2(\mathbb{F}_{k})} = e^{i s \cdot k}
\eeq
if 
\beq
	\mathbb{F}_{k} := \{ (k_{1}, k-k_{1}) \in \mathbb{T}^2 \times \mathbb{T}^2 | k_{1} \in \mathbb{T}^2 \} \cong  \mathbb{T}^2
\eeq
for all $k \in  \mathbb{T}^2$. Therefore, $L^2(\mathbb{F}_{k})$ is the isotypical component of the 1-dimensional irreducible representation of the group of simultaneous translations of the two particles with character $e^{i s \cdot k}$, $s \in \mathbb{Z}^2$. The operator $\vee \circ T \circ \wedge$ is decomposable with respect to the fibration 
\beq
	L^2(\mathbb{T}^2 \times \mathbb{T}^2) = \int_{\mathbb{T}^2} \oplus L^2(\mathbb{F}_{k}) dk
\eeq
because $[\vee \circ T \circ \wedge, U_{s}] = 0$ for all $s \in \mathbb{Z}^2$. Therefore, $T$ itself is decomposable with respect to the fibration 
\beq
	l^{2}(\mathbb{Z}^2 \times \mathbb{Z}^2) = \int_{\mathbb{T}^2} \oplus \wedge\left(L^2(\mathbb{F}_{k})\right) dk.
\eeq 
We can express this observation in the form
\beq
	T = \int_{\mathbb{T}^2} \oplus T(k) dk
\eeq
Note that
\beq
	\sigma(T) = \bigcup_{k \in \mathbb{T}^2} \sigma(T(k))
\eeq
Thus, to prove our claim $\sigma_{ac}(T) \neq \emptyset$ from above it suffices to show that $\sigma_{ac}(T(k)) \neq \emptyset$ for some $k \in \mathbb{T}^2$. To figure out the spectra of the operators $T(k)$ we need to compute the operators $T(k)$ more explicitly by the determination of their matrix elements with respect to a family of wave functions that generate the spaces $\wedge\left(L^2(\mathbb{F}_{k})\right) \subset l^{2}(\mathbb{Z}^2 \times \mathbb{Z}^2)$.

\begin{lemma}
	Let $k \in \mathbb{T}^2$ be arbitrary. The set $\hat{\mathcal{B}_{k}} := \left.\left\{  \hat{e}_{j}^{(k)} \right| j \in \mathbb{Z}^2 \right\}$ with
	\beq\label{pf.lemma.ONB.fiber}
		\hat{e}_{j}^{(k)}(n_{1},n_{2}) := \frac{\sqrt{2}}{(2\pi)^2} e^{-i n_{2} \cdot k} \delta_{j, (n_{1} - n_{2})}
	\eeq 
	($n_{1}, n_{2} \in \mathbb{Z}^2$) forms a complete orthonormal basis in $\wedge\left(L^2(\mathbb{F}_{k})\right) \subset l^{2}(\mathbb{Z}^2 \times \mathbb{Z}^2)$.
\end{lemma}

\begin{proof}
	The family $\{ e^{i j_{1} q_{1}} e^{i j_{2} q_{2}} | j_{1},j_{2} \in \mathbb{Z}^2 \}$ of functions ($q_{1}, q_{2} \in \mathbb{T}$) form the standard basis in $L^2(\mathbb{T}^2)$. The map $u_{k} : L^2(\mathbb{F}_{k}) \rightarrow L^2(\mathbb{T}^2)$ defined by
	\beq
		(u_{k} \psi)(q) := \sqrt{2} \, \psi(q, k-q)
	\eeq
	($q \in \mathbb{T}^2$) is unitary, i.e.,
	\beq
		\sp{\psi_{1}}{\psi_{2}}_{L^2(\mathbb{F}_{k})} = \sp{u_{k} \psi_{1}}{u_{k} \psi_{2}}_{L^2(\mathbb{T}^2)}.
	\eeq
	Here, 
	\beq
		\sp{\psi_{1}}{\psi_{2}}_{L^2(\mathbb{F}_{k})} = \int_{[0,2\pi)^4} d^4 \, \mathbf{k} \bar{\psi}_{1}(\mathbf{k}) \chi_{\mathbb{F}_{k}}(\mathbf{k})  \psi_{2}(\mathbf{k}) \chi_{\mathbb{F}_{k}}(\mathbf{k})
	\eeq
	and $\chi_{\mathbb{F}_{k}}(\cdot)$ denotes the characteristic function with respect to $\mathbb{F}_{k} \subset \mathbb{T}^4$. Thus $\{ u_{k}^{-1}  e^{i j_{1} (\cdot)} e^{i j_{2} (\cdot)} | j_{1},j_{2} \in \mathbb{Z}^2  \}$ is a complete orthonormal basis for $L^2(\mathbb{F}_{k})$. These observations allow for the explicit computation of the wanted basis
	\beq
		\hat{\mathcal{B}_{k}} := \left.\left\{          \left( \wedge \circ u_{k}^{-1}  e^{i j_{1} (\cdot) } e^{i j_{2} (\cdot)} \right) (n_{1},n_{2})       \right| j \in \mathbb{Z}^2 \right\}.
	\eeq
	The calculation yields the functions in Eq.~\eqref{pf.lemma.ONB.fiber}.
\end{proof}

We continue with the computation of the matrix elements
\beq
	T(k)(l,j) = \sp{\hat{e}_{l}^{(k)}}{T(k) \hat{e}_{j}^{(k)}} = \sp{\hat{e}_{l}^{(k)}}{T \hat{e}_{j}^{(k)}}
\eeq
($l,j \in \mathbb{Z}^2, k \in \mathbb{T}^2$) of $T(k)$ with respect to the basis $\hat{\mathcal{B}_{k}}$ in $\wedge\left(L^2(\mathbb{F}_{k})\right) \subset l^{2}(\mathbb{Z}^2 \times \mathbb{Z}^2)$. Note that 
\beq\label{support.basis.for.relative.dynamics}
	\supp{\hat{e}_{j}^{(k)}(\cdot, \cdot)} = \{  (n_{1},n_{2}) \in \mathbb{Z}^2 \times \mathbb{Z}^2 | j = n_{1} - n_{2})  \}.
\eeq
From the perturbation theory stems the constraint that the particles are not allowed to sit on top of each other because this would lead to an annihilation of the two particles resulting in a departure of the state vector from the unperturbed eigenspace. We can realize this boundary condition dynamically by setting all matrix elements of $T$ that would lead to double occupations of lattice sites to zero. From~\eqref{support.basis.for.relative.dynamics} we conclude that the forbidden subspace is spanned by $\{ \hat{e}_{0}^{(k)} |  k \in \mathbb{T}^2 \}$. To dynamically realize the boundary condition we thus set
\beq
	T(k)(l,j) = 0
\eeq
whenever $l=0$ or $j=0$. We are thus left with the computation of the matrix elements
\beq
	T(k)(l,j) = \sp{\hat{e}_{l}^{(k)}}{T \hat{e}_{j}^{(k)}}
\eeq
for $l \neq 0$ and $j \neq 0$. Let $\tilde{T}_{0}$ be the translationally invariant and self-adjoint operator (arbitrary translations in $\mathbb{Z}^4$) on $l^{2}(\mathbb{Z}^2 \times \mathbb{Z}^2)$ with the property
\beq
	\tilde{T}_{0}(n_{1},n_{2},m_{1},m_{2}) = T(n_{1},n_{2},m_{1},m_{2})
\eeq
($n_{1},n_{2},m_{1},m_{2} \in \mathbb{Z}^2$) whenever $\| n_{1} - n_{2} \|_{2} \geq R+m+1$ ($R$ denotes the spatial cutoff of the spin-spin interaction; recall that $\max\{ \| i_{k} - i_{l} \| \, : \, k,l \in \{ 1,...,m \}  \} > R$ implies zero interaction). Note that $\tilde{T}_{0}$ is well-defined because there are no interactions (between the two particles) within the constraining domain  
\beq\label{Constraining.domain}
	\{ (n_{1},n_{2}) \in \mathbb{Z}^4 \, | \,  \| n_{1} - n_{2} \|_{2} \geq R+m+1  \}.
\eeq
Define
\beq
	\tilde{T}_{I} := T - \tilde{T}_{0}.
\eeq
We conclude that $\tilde{T}_{I}(\cdot.\cdot)$ is only supported on a finite neighborhood of the origin in $\mathbb{Z}^4$ and that $\tilde{T}_{0}$ and $\tilde{T}_{I}$ are of the form
\begin{align}
	\left(\tilde{T}_{0}  \hat{e}_{j}^{(k)}\right)(\mathbf{n}) &= \sum_{\mathbf{p} \in \mathcal{P}} \xi^{(0)}_{\mathbf{p}} \, \hat{e}_{j}^{(k)}(\mathbf{n} - \mathbf{p}) \nn \\
	\left(\tilde{T}_{I}  \hat{e}_{j}^{(k)}\right)(\mathbf{n}) &= \sum_{\mathbf{p} \in \mathcal{P}} \xi^{(I)}_{\mathbf{p}}(n_{1} - n_{2}) \, \hat{e}_{j}^{(k)}(\mathbf{n} - \mathbf{p})
\end{align}
($\xi^{(0)}$ is by construction translationally invariant while $\xi^{(I)}_{\mathbf{p}}(n_{1} - n_{2})$ is not). The explicit form of the constraining domain~\eqref{Constraining.domain} implies
\beq\label{where.xi.I.vanishes}
	\xi^{I}_{\mathbf{p}}(n_{1} - n_{2}) = 0
\eeq
whenever $\| n_{1} - n_{2} \|_{2} \geq R+m+1$. The set $\mathcal{P} \subset \mathbb{Z}^4$ is contained in $\{ \mathbf{a} \in \mathbb{Z}^4 | \| \mathbf{a} \|_{\infty} \leq m +1 \}$ and the coefficients $\xi_{\mathbf{p}}^{(0)} + \xi_{\mathbf{p}}^{(I)}$ are equal to the sum of those coupling constants $\xi(i_{1},z,...,i_{m},z)$ with the property that the difference between the endpoint and the start-point of the path that corresponds to $\{ i_{1},...,i_{m} \}$ equals $\mathbf{p} \in \mathbb{Z}^4$. Thus,
\begin{align}
	T(k)(l,j)	
	=&	\sp{\hat{e}_{l}^{(k)}}{T \hat{e}_{j}^{(k)}} \nn \\
	=&	\sp{\hat{e}_{l}^{(k)}}{(\tilde{T}_{0} + \tilde{T}_{I}) \hat{e}_{j}^{(k)}} \nn \\
	=&	\sum_{(n_{1},n_{2}) \in \mathbb{Z}^4}	\overline{\hat{e}_{l}^{(k)}((n_{1},n_{2}))} \left( \sum_{\mathbf{p} \in \mathcal{P}} \left(\xi^{(0)}_{\mathbf{p}} + \xi^{I}_{\mathbf{p}}(n_{1} - n_{2})\right) \, \hat{e}_{j}^{(k)}((n_{1},n_{2}) - \mathbf{p}) \right) \nn \\
	=&	\frac{2}{(2\pi)^4} \sum_{(n_{1},n_{2}) \in \mathbb{Z}^4}  e^{i n_{2} \cdot k} \delta_{l, (n_{1} - n_{2})}   \left( \sum_{\mathbf{p} \in \mathcal{P}} \left(\xi^{(0)}_{\mathbf{p}} + \xi^{I}_{\mathbf{p}}(n_{1} - n_{2})\right) \, e^{-i (n_{2} - p_{2}) \cdot k} \delta_{j, (n_{1} - n_{2} + p_{2} -p_{1})} \right).
\end{align}
The substitutions $n_{r} := n_{1} - n_{2}$ and $n_{a} := n_{2}$ lead to
\begin{align}
	T(k)(l,j)	
	=&	\frac{2}{(2\pi)^4} \sum_{(n_{r},n_{a}) \in \mathbb{Z}^4}  e^{i n_{a} \cdot k} \delta_{l, n_{r}}   \left( \sum_{\mathbf{p} \in \mathcal{P}} \left(\xi^{(0)}_{\mathbf{p}} + \xi^{I}_{\mathbf{p}}(n_{r})\right) \, e^{-i (n_{a}-p_{2}) \cdot k} \delta_{j, n_{r} + p_{2} - p_{1}} \right) \nn \\
	=&	\frac{2}{(2\pi)^4} \left( \sum_{n_{a} \in \mathbb{Z}^2}  e^{i n_{a} \cdot k} e^{-i n_{a} \cdot k} \right) \left( \sum_{n_{r} \in \mathbb{Z}^2} \sum_{\mathbf{p} \in \mathcal{P}}  \delta_{l, n_{r}}   \left(\xi^{(0)}_{\mathbf{p}} + \xi^{I}_{\mathbf{p}}(n_{r})\right) \, e^{i p_{2} \cdot k}  \delta_{j, n_{r} + p_{2} - p_{1}} \right) \nn \\
	=&	\frac{2}{(2\pi)^4} \left( \sum_{n_{a} \in \mathbb{Z}^2}  e^{i n_{a} \cdot k} e^{-i n_{a} \cdot k} \right) \left( \sum_{\mathbf{p} \in \mathcal{P}}  \left(\xi^{(0)}_{\mathbf{p}} + \xi^{I}_{\mathbf{p}}(l)\right) \, e^{i p_{2} \cdot k}  \delta_{j, l + p_{2} - p_{1}} \right)
\end{align}
The divergent first factor is the scattering part that we ignore by setting it equal to 1 because we are interested in the relative dynamics of the two particles. This divergent part also appears in the two-body problems in the continuum; it corresponds to the free evolution of the center of mass. We continue by reinterpreting the space $\mathbb{Z}^2$ that carries the labels $l$ and $j$ as the configuration space associated to the relative dynamics of the two particles. We do this for each $k \in \mathbb{T}^2$ separately. The relative dynamics is thus generated by the hopping Hamiltonian $T(k)_{0} + T(k)_{I}$ where
\begin{align}
	T(k)_{0}(l,j)	&:=	\frac{2}{(2\pi)^4}  \sum_{\mathbf{p} \in \mathcal{P}}   \xi^{(0)}_{\mathbf{p}}  \, e^{i p_{2} \cdot k}  \delta_{j, l + p_{2} - p_{1}}    \label{Complete.relative.dynamics1} \\
	T(k)_{I}(l,j)	&:=	\frac{2}{(2\pi)^4}  \sum_{\mathbf{p} \in \mathcal{P}}   \xi^{(I)}_{\mathbf{p}}(l)  \, e^{i p_{2} \cdot k}  \delta_{j, l + p_{2} - p_{1}}  \label{Complete.relative.dynamics2}
\end{align}
Note that $T(k)_{I}(l,j)$ vanishes whenever $l \geq R+m+1$ or $j \geq R+m+1$. This is a consequence of~\eqref{where.xi.I.vanishes} and the self-adjointness of $T(k)$. At infinity the relative dynamics is thus generated by translationally invariant and finite-ranged hopping (see~\eqref{Complete.relative.dynamics1}).

\subsection{Existence of Nonempty Absolutely Continuous Spectrum}

In the present subsection we show that the absolutely continuous spectrum $\sigma_{ac}(T(k)_{0} + T(k)_{I})$ of the family of operators describing the relative motion of the pair of defects is nonempty. To that purpose we first show that $\sigma(T(k)_{0}) =\sigma_{ac}(T(k)_{0})$ and use the Kato Rosenblum Theorem (that is stated below) afterwards to infer 
\beq
	\emptyset \neq \sigma_{ac}(T(k)_{0})  \subseteq \sigma_{ac}(T(k)_{0} + T(k)_{I}).
\eeq

\begin{lemma}\label{Existence.of.scattering.states.Lemma0}
	The spectrum of $T(k)_{0}$ is absolutely continuous, i.e.,
	\beq
		\sigma(T(k)_{0}) =\sigma_{ac}(T(k)_{0}).
	\eeq
\end{lemma}

\begin{proof}

Fix $k \in [0,2\pi)^2$ and set $H := T(k)_{0}$. According to the spectral theorem there exists a projection valued measure $(P_{\Delta})_{\Delta}$ such that 
\beq\label{temp7348274}
	H = \int_{\AR} r \, dP_{r}. 
\eeq
To prove the assertion we have to show that for every $\hat{\psi} \in l^2(\mathbb{Z}^2)$ there exists a function $m_{\hat{\psi}} \in L^1(\AR)$ with the property
\beq\label{temp.jdisjdiae}
	\mu_{\hat{\psi}}(\Delta) := \sp{\hat{\psi}}{P_{\Delta} \hat{\psi}} = \sp{\hat{\psi}}{\int_{\Delta} dP_{r}(\hat{\psi})} = \int_{\Delta} m_{\hat{\psi}}(r) \, dr
\eeq
($\Delta \subseteq \AR$). To determine $\int_{\Delta} dP_{r}(\hat{\psi})$ we have a closer look at $H \hat{\psi}$:
\begin{align*}
	(H \hat{\psi})(n)
	&=	\left( \frac{1}{2\pi} \right)^2 \int_{[0,2\pi)^2} \psi(q) \left( H e^{-i q \cdot (\cdot)}\right)(n) \, dq \\
	&=	\left( \frac{1}{2\pi} \right)^2 \int_{[0,2\pi)^2} \psi(q) \left( E_{q} e^{-i q \cdot (\cdot)}\right)(n) \, dq \\
	&=	\left( \frac{1}{2\pi} \right)^2 \int_{\{ (r,\phi) | q(r,\phi) \in [0,2\pi)^2 \}} \psi(q(r,\phi))  E_{q(r,\phi)} e^{-i q(r,\phi) \cdot n} \, d\phi \, rdr \\
	&=	\int_{0}^{2\sqrt{2}\pi} dr \, r \int_{\{ \phi(r) \}}	\left( \frac{1}{2\pi} \right)^2  \psi(q(r,\phi))  E_{q(r,\phi)} e^{-i q(r,\phi) \cdot n}  d\phi
\end{align*}
($(r,\phi) \in \AR_{+} \times [0,2\pi)$ are polar coordinates) with
\beq
	E_{q} = \frac{2}{(2\pi)^4}  \sum_{\mathbf{p} \in \mathcal{P}}   \xi^{(0)}_{\mathbf{p}}  \, e^{i p_{2} \cdot k} e^{i q\cdot (p_{2} - p_{1})}
\eeq
because
\begin{align}\label{existence.of.scattering.states.pf.Lemma1.1}
		(T(k)_{0} e^{i q \cdot (\cdot)})(n)	
		&=	\sum_{m \in \mathbb{Z}^2}  T(k)_{0}(n,m) e^{i q \cdot m} \nn \\
		&=	\sum_{m \in \mathbb{Z}^2}  \frac{2}{(2\pi)^4}  \sum_{\mathbf{p} \in \mathcal{P}}   \xi^{(0)}_{\mathbf{p}}  \, e^{i p_{2} \cdot k}  \delta_{m, n + p_{2} - p_{1}} e^{i q\cdot m} \nn \\
		&=	\frac{2}{(2\pi)^4}  \sum_{\mathbf{p} \in \mathcal{P}}   \xi^{(0)}_{\mathbf{p}}  \, e^{i p_{2} \cdot k} e^{i q\cdot (n + p_{2} - p_{1})} \nn \\
		&=	\left( \frac{2}{(2\pi)^4}  \sum_{\mathbf{p} \in \mathcal{P}}   \xi^{(0)}_{\mathbf{p}}  \, e^{i p_{2} \cdot k} e^{i q\cdot (p_{2} - p_{1})} \right)  e^{i q \cdot n}.
\end{align}
From Eq.~\eqref{temp7348274} we thus conclude
\beq
	\int_{\Delta} dP_{r}(\hat{\psi}) =  \int_{\Delta}  \chi_{[0,2\sqrt{2}\pi]}(r)    \int_{\{ \phi(r) \}}	\left( \frac{1}{2\pi} \right)^2  \psi(q(r,\phi))  E_{q(r,\phi)} e^{-i q(r,\phi) \cdot n}  d\phi     dr
\eeq
($\Delta \subseteq \AR$). Now we can go back to Eq.~\eqref{temp.jdisjdiae}:
\begin{align}
	\sp{\hat{\psi}}{\int_{\Delta} dP_{r}(\hat{\psi})}
	&=	\sum_{n \in \mathbb{Z}^2} \overline{\hat{\psi}(n)}  \left( \int_{\Delta} dP_{r}(\hat{\psi})  \right)(n)  \\
	&=	\int_{\Delta}  \chi_{[0,2\sqrt{2} \pi]}(r)    \int_{\{ \phi(r) \}} \left( \frac{1}{2\pi} \right)^2  \psi(q(r,\phi))  E_{q(r,\phi)}      \overline{ \left( \sum_{n \in \mathbb{Z}^2} \hat{\psi}(n)   e^{i q(r,\phi) \cdot n} \right)  }     d\phi     dr .
\end{align}
Consequently,
\beq
	\sp{\hat{\psi}}{\int_{\Delta} dP_{r}(\hat{\psi})} =  \int_{\Delta}  \chi_{[0,2\sqrt{2} \pi]}(r)    \int_{\{ \phi(r) \}} \left( \frac{1}{2\pi} \right)^2    E_{q(r,\phi)}  | \psi(q(r,\phi)) |^2     d\phi     dr 
\eeq
because the bracket equals the inverse discrete Fourier transform. The definition
\beq
	m_{\hat{\psi}}(r) := \chi_{[0,2\sqrt{2} \pi]}(r)    \int_{\{ \phi(r) \}} \left( \frac{1}{2\pi} \right)^2    E_{q(r,\phi)}  | \psi(q(r,\phi)) |^2     d\phi \in L^1(\AR)
\eeq
concludes the proof of the Lemma.

\end{proof}

One says that the \emph{generalized wave operators $\Omega^{\pm}(A,B)$ exist} if the strong limits
\beq
	\Omega^{\pm}(A,B) = \mathrm{s}-\lim_{t\rightarrow \mp \infty} e^{iAt} e^{-iBt} P_{\mathrm{ac}}(B)
\eeq
exist (see~Ref.~\onlinecite{Reed.Simon.III}). The operator $A \upharpoonright \mathrm{Ran} \, \Omega^{\pm}(A,B)$ is unitarily equivalent to $B \upharpoonright [P_{\mathrm{ac}}(B) l^2(\mathbb{Z}^2)]$ if $\Omega^{\pm}(A,B)$ exist (see the proof of part (c) of Proposition~1 in~Ref.~\onlinecite{Reed.Simon.III}). In our case, the Kato~Rosenblum Theorem guarantees the existence of the operators $\Omega^{\pm}(T(k)_{0} + T(k)_{I},T(k)_{0})$ (recall that $T(k)_{I}$ is trace class):

\begin{theorem}[Kato Rosenblum, see Ref.~\onlinecite{Reed.Simon.III}]
	Let $A$ and $B$ be two self-adjoint operators with $A-B$ being trace class. Then, $\Omega^{\pm}(A,B)$ exist and are complete.
\end{theorem}

We conclude that $T(k)_{0} \upharpoonright [P_{\mathrm{ac}}(T(k)_{0})  l^2(\mathbb{Z}^2)] = T(k)_{0}$ (see Lemma~\ref{Existence.of.scattering.states.Lemma0}) is unitarily equivalent to $T(k)_{0} + T(k)_{I} \upharpoonright \mathrm{Ran} \, \Omega^{\pm}(T(k)_{0} + T(k)_{I},T(k)_{0})$ and therefore
\begin{align*}
	\sigma\left(T(k)_{0} + T(k)_{I} \upharpoonright \mathrm{Ran}  \, \Omega^{\pm}(T(k)_{0} + T(k)_{I},T(k)_{0})\right) 
	&=	\sigma_{ac}\left(T(k)_{0} + T(k)_{I} \upharpoonright \mathrm{Ran}  \, \Omega^{\pm}(T(k)_{0} + T(k)_{I},T(k)_{0})\right) \\
	&=	\sigma_{ac}\left(T(k)_{0}\right) \neq \emptyset
\end{align*}
(use Lemma~\ref{Existence.of.scattering.states.Lemma0}). This proves the following Theorem because
\beq
	\sigma_{ac}\left(T(k)_{0} + T(k)_{I} \upharpoonright \mathrm{Ran} \, \Omega^{\pm}(T(k)_{0} + T(k)_{I},T(k)_{0})) \subseteq \sigma_{ac}(T(k)_{0} + T(k)_{I}\right).
\eeq

\begin{theorem}
	The absolutely continuous spectrum of $T(k)_{0} + T(k)_{I}$ is nonempty.
\end{theorem}

\section{One-Particle Localization}\label{One.Particle.Localization}

From the work emanating from Anderson's discovery from 1958 we know the following: Let $H_{0}$ be a one-particle Hamiltonian on the lattice $\mathbb{Z}^d$ with finite-ranged hopping and assume that $\psi$ is some scattering state of $H_{0}$. Assume $H_{0}$ is perturbed by a random potential $\lambda V$. Then --- when turning on $\lambda$ ---  all the eigenvectors and generalized eigenvectors of the unperturbed Hamiltonian that are  associated to some energy interval $I \subset \AR$ turn (under some circumstances) into exponentially decaying eigenvectors i.e., there exist $A_{n}$ such that
\beq
	| \phi_{n}(x) |^2 \leq A_{n} e^{-\mu_{n} | x |}
\eeq
for all eigenvectors $\phi_{n}$ of $H_{0} + \lambda V$ with eigenvalue $E \in I$. This effect is called \emph{spectral localization in $I \subset \AR$}. Another established notion is \emph{dynamical localization in $I \subset \AR$}: Every initially localized wave function with spectral components in $I \cap \sigma(H_{0} + \lambda V)$ will remain exponentially localized at all times. A sufficient condition for dynamical localization is 
\beq\label{dyn.loc}
	\EW{ \sup_{t\in \AR}  | \sp{\delta_{y}}{e^{-iHt} P_{I}(H) \delta_{x}} | } \leq A e^{-|x-y|/\xi}
\eeq
for some $A,\xi \in (0,\infty)$ which depend themselves on $I$ ($P_{I}(H)$ is the spectral projector of $H$ with respect to $I$). The notation $\EW{...}$ denotes averaging with respect to different realizations of the random potential $\lambda V$. Note that this is exactly what we are looking for: The expected probability for measuring the initial condition (i.e., $\delta_{x}$) at $y$ decays --- independent of the time $t$ --- exponentially in the distance between $x$ and $y$. From a naive point of view one could expect that the notions ``dynamical localization'' and ``spectral localization'' are equivalent. In 1995, del Rio et al. showed that this is not true in general: while dynamical localization implies spectral localization the opposite is sometimes violated~\cite{Simon1995}. However in 1994 Aizenman proved a more general version of the following sufficient criterion for dynamical localization (cf. Ref.~\onlinecite{Aizenman1994}).

\begin{theorem}[see Ref.~\onlinecite{Aizenman1994}]\label{1p.FMM}
	Let $H = T + V(x)$ be a Hamiltonian acting on the Hilbert space $l^2(\mathbb{Z}^d)$ with $T$ being a finite-range hopping operator and let $H_{\Omega}$, $\Omega \subset \mathbb{Z}^d$, be obtained from $H$ by setting to zero all the hopping terms starting of and ending at elements outside of $\Omega$. Assume that the values $V(x)$, $x \in \mathbb{Z}^d$, building up the potential are iid random variables with a probability measure given in terms of a compactly supported and bounded density $\rho(v)$. Then,
	\beq\label{aizenman.criterium}
		\EW{| \sp{\delta_{y}}{(H_{\Omega} - E)^{-1} \delta_{x}} |^s} \leq C_{s} \, e^{- \mu | x-y |}
	\eeq
	for all $E \in I$ with some $\Omega$-independent constants $s \in (0,1)$, $\mu > 0$, $C < \infty$ forms a sufficient condition for dynamical localization~\eqref{dyn.loc}.
\end{theorem}

The method to prove dynamical localization via the criterium~\eqref{aizenman.criterium} is sometimes called Fractional Moment Method (FMM). The criterium~\eqref{aizenman.criterium} is satisfied for a large class of physical systems in any finite dimension including systems with so called ``high disorder'':

\begin{theorem}[see Ref.~\onlinecite{Aizenman1993}]\label{1p.satisfaction.of.FMM}
	Let $H, H_{\Omega}$ be as above. Then there exists a $\lambda_{0} < \infty$ such that for all $\lambda > \lambda_{0}$ and all energies $E$
	\beq
		\EW{| \sp{\delta_{y}}{(H_{\Omega} - E)^{-1} \delta_{x}} |^s}    \leq   C_{s} e^{- \mu |x-y|}
	\eeq
	with $\mu > 0$, $C_{s} < \infty$ and $| ... |$ denotes (for example) the 1-norm on $\mathbb{Z}^d$.
\end{theorem}

This theorem is a special case of Lemma 3.2 in~Ref.~\onlinecite{Aizenman1993} by Aizenman and Molchanov. Other proofs of localization are based on the so called multiscale analysis (MSA) approach invented by Fr\"ohlich and Spencer~\cite{FrohlichSpencer1983}. We conculde that if the perturbation $\lambda V$ leads to the satisfaction of the criterium \eqref{aizenman.criterium} then the expectation value of the amplitude for measuring a particle (that has been initially compactly supported) outside some ball with radius $R$ goes exponentially to zero with increasing $R$.


\section{Many-Particle Localization}\label{Many.Particle.Localization}

Note that the 2-particle system with state space $\Eig{H_{TC}}{E_{1,0}}$ and dynamics $P_{1,0} H_{I} P_{1,0}$ is an \emph{interacting} two-particle system so that we are actually not allowed to blindly apply the 1-particle theorems from before. Luckily there has been a lot of progress in the investigation of localization in interacting many-body systems. In 2009 Aizenman and Warzel~\cite{AizenmanWarzel2009} proved dynamical localization of interacting $n$-body systems ($n < \infty$) on $\mathbb{Z}^d$ under the assumption that the interactions are described by interaction potentials with finite range. More precisely, they considered $n$-particle Hamiltonians with up to $m$-point interactions ($m < \infty$) of the form
\beq\label{Aizenman.Warzel:Hamiltonian}
	H^{(n)} = \sum_{j=1}^n \left[ -\Delta_{j} + \lambda V(x_{j}) \right] + \mathcal{U}(\mathbf{x};\mathbf{\alpha})
\eeq
($\mathbf{x} = (x_{1},x_{2},...,x_{n}) \in \mathbb{Z}^{dn}$) where $\Delta_{j}$, $V(x)$ and $\mathcal{U}(\mathbf{x};\mathbf{\alpha})$ denote discrete Laplacian, iid random one-particle potential and interaction potential respectively (cf.~Ref.~\onlinecite{AizenmanWarzel2009} for the precise specifications) and prove a more general version of the following theorem.

\begin{theorem}[see Ref.~\onlinecite{AizenmanWarzel2009}]\label{Aizenman2009theorem}
	For each $n \in \mathbb{N}$ and $m \in \{ 1, ..., n \}$ there is an open set $\mathcal{L}_{n}^{(m)} \subset \mathbb{R}_{+} \times \mathbb{R}^m$ which includes regimes of strong disorder and weak interactions (see Ref.~\onlinecite{AizenmanWarzel2009} for the precise characterization of these regimes), for which at some $A,\xi < \infty$ and all $(\lambda, \mathbf{\alpha}) \in \mathcal{L}_{n}^{(m)}$, and all $\mathbf{x}, \mathbf{y} \in \mathbb{Z}^{nd}$
	\beq
		\mathbb{E}\left[ \sup_{t \in \AR} \large|  \langle  \delta_{\mathbf{x}} , e^{-it H^{(n)}}  \delta_{\mathbf{y}}      \rangle   \large|  \right] \leq A \, e^{-\text{dist}_{\mathcal{H}}(\mathbf{x}, \mathbf{y}) / \xi  } ,
	\eeq
	where the so called Hausdorff pseudo distance is defined by
	\beq
		\text{dist}_{\mathcal{H}}(\mathbf{x}, \mathbf{y}) := \max \left\{    \max_{1\leq i \leq k} \text{dist}(x_{i}, \{ \mathbf{y} \} ) ,  \max_{1\leq i \leq k} \text{dist}(\{ \mathbf{x} \}, y_{i} )    \right\}.
	\eeq
\end{theorem}

As in the 1-particle case (see Theorems~\ref{1p.FMM} and~\ref{1p.satisfaction.of.FMM}), Aizenman and Warzel prove Theorem~\ref{Aizenman2009theorem} by deriving first the many-body analog of the sufficient 1-particle criterium~\ref{1p.FMM} and by showing the validity of the validity of this criterium afterwards. Thus, the first step in the proof of Theorem~\ref{Aizenman2009theorem} is the generalization of the sufficient criterium described in Theorem~\ref{1p.FMM} to interacting many-body systems:

\begin{theorem}[Fractional Moment Criterium, see Ref.~\onlinecite{AizenmanWarzel2009}]\label{np.FMM}
	Let $H^{(n)}$ be the Hamiltonian~\eqref{Aizenman.Warzel:Hamiltonian} acting on the many-body Hilbert space $l^2(\mathbb{Z}^{nd})$, $I \subseteq \AR$, and let $H^{(n)}_{\Omega}$ ($\Omega \subset \mathbb{Z}^d$) be the finite-volume operator that is obtained from $H^{(n)}$ by keeping all matrix elements that map $C^{(n)}(\Omega)$ to itself unchanged and setting all other matrix elements to zero. Assume that the values $V(x)$, $x \in \mathbb{Z}^d$, building up the background potential are iid random variables with probability measure given in terms of a compactly supported and bounded density $\rho(v)$. Then the following is a sufficient criterium for dynamical localization (with respect to the Hausdorff pseudo-distance) for energies within the interval $I$: There exist $A, \xi < \infty$ and $s \in (0,1)$ such that
	\beq\label{np.aizenman.criterium}
		\sup_{\substack{ I \subset \AR \\ | I | \geq 1  }} \sup_{\Omega \in \mathbb{Z}^d} \frac{1}{| I |} \int_{I}  \EW{  |  G_{\Omega}^{(n)}(\mathbf{x},\mathbf{y};E)  |^s }  \, dE  \leq  A e^{-\frac{\distH{x}{y}}{\xi}}   
	\eeq
	The existence of exponential bounds for fractional moments of finite-volume Green's functions (cf.~\eqref{np.aizenman.criterium}) is equivalent to the existence of exponential bounds for finite-volume eigenfunction corrolators, i.e.,
	\beq\label{np.aizenman.criterium.B}
		\sup_{\substack{ I \subset \AR   }} \sup_{\Omega \in \mathbb{Z}^d} \EW{  Q_{\Omega}^{(n)}(\mathbf{x}, \mathbf{y}; I)  }    \leq  A e^{-\frac{\distH{x}{y}}{\xi}}.
	\eeq
\end{theorem}

Here,
\beq
	Q_{\Omega}^{(n)}(\mathbf{x}, \mathbf{y}; I) := \sum_{E \in \sigma(H_{\Omega}^{(n)}) \cap I }   \left|  \sp{ \delta_{\mathbf{x}} }{  P_{\{ E \} }(H_{\Omega}^{(n)}) \delta_{\mathbf{y}} }  \right|.
\eeq
In the remainder we will refer to the sufficient criterium from Theorem~\ref{np.FMM} in terms of ``fraction moment criterium''.

\section{Proof of the $n$-particle formulation of Theorem~2}\label{Proof.of.Theorem.2}

Our goal is to use dynamical localization to stop the spreading of defect-wave functions that we have encountered at the example of relative 2-defect propagation in the section labelled ``Propagation to Infinity''. Theorem~\ref{Aizenman2009theorem} does not immediately lead to dynamical localization because in our setting the finite-range interactions between the two defects that are evolving according to $P_{1,0} H_{I} P_{1,0}$ are given in terms of inhomogeneous hopping matrix elements. Thus, in order to get a localization bound for the many-defect system under consideration we need to go through the proof of Theorem~\ref{Aizenman2009theorem} and adapt it to the setting described in Theorem~2. The following proof of Theorem~2 not only covers the 2-particle cases that are described in the main text but also general $n$-particle cases ($n \in \mathbb{N}$). We are allowed to set $\alpha = 0$ (see Eq.~\eqref{Aizenman.Warzel:Hamiltonian}) because we are not dealing with interactions that are specified in terms of interaction potentials.

In Ref.~\onlinecite{AizenmanWarzel2009} Theorem~2.1, Lemma~3.1, Theorem~4.1, Theorem~4.2, Lemma~4.3, Lemma~4.4, Theorem~4.5 and Lemma~5.1 lead to the proof of the many-body version of the fractional moment criterium (cf. Theorem~\ref{np.FMM}). To prove that~\eqref{np.aizenman.criterium} and~\eqref{np.aizenman.criterium.B} still serve as a sufficient criterium for dynamical localization in our setup (cf. Theorem~2) we only need to adapt the proof of Lemma~4.6 in~Ref.~\onlinecite{AizenmanWarzel2009}:

\begin{lemma}[Lemma 4.6]
	Let $\Omega \subset \mathbb{Z}^d$ and $E \geq 0$. Then for every $\mathbf{x} \in \mathcal{C}^{(n)}(\Omega)$:
	\beq\label{eq1.lemma4.6}
		\EW{  \sp{\deltafunc{x}}{P_{\AR \backslash (-E,E)}(H_{\Omega}^{(n)}) \deltafunc{x}  }  }   \leq  \EW{    e^{ V(0) }   } e^{\min\{ 1,(n | \lambda |)^{-1} \} (| \max \mathrm{supp} H_{0} | \, h_{\max}   -E)    },
	\eeq
	with $h_{\max} := \max_{i_{1},...,i_{m} \in \mathcal{B}, i_{q} \neq i_{t}}  | \xi(i_{1},z,...,i_{m},z) |$ and $\max \mathrm{supp} H_{0} := \bigcup_{\mathbf{k}} \left[ \mathrm{supp} H_{0}(\mathbf{k},\cdot) - \mathbf{k} \right]$ denotes an upper bound on the $\mathbf{x}$-dependent support of $H_{0}(\mathbf{x}, \cdot)$.
\end{lemma}

Note that $| \max \mathrm{supp} H_{0} | < \infty$ because $| \max \mathrm{supp} H_{0} |$ is determined by the $(i_{1},...,i_{m})$-movements in $\mathbb{Z}^{nd}$ and thus $| \max \mathrm{supp} H_{0} |$ is bounded by the number of points in the ball $B_{m}(0)$ with 1-norm radius $m$ located at the origin. 

\begin{proof}
	In Ref.~\onlinecite{AizenmanWarzel2009} the authors used the Chebyshev-type inequality $1_{\AR \backslash (-E,E)}(x) \leq e^{-tE} (e^{tx} + e^{-tx})$ to upper bound the LHS of~\eqref{eq1.lemma4.6} by semigroups. To find convenient upper bounds we continue along the lines in section 3 in~Ref.~\onlinecite{Frohlich2006} using Duhamel's formula
	\beq
		e^{-t(H_{0} + V)} = e^{-t V} + \int_{0}^{t} d\tau \, e^{-\tau V} (-H_{0}) e^{-(t - \tau)(H_{0} + V)}  
	\eeq 
	where $H_{0}$ and $V$ denote deterministic finite ranged hopping and random potential respectively. The iteration of Duhamel's formula gives
	\beq\label{temp.32372837}
		e^{-t(H_{0} + V)} = \sum_{m \geq 0} \int_{0 < \tau_{1} < ... < \tau_{m} < t} d\tau_{1} \cdots d\tau_{m} \, e^{-\tau_{1} V} (-H_{0}) e^{-(\tau_{2} - \tau_{1}) V} (-H_{0}) \cdots (-H_{0}) e^{-(t - \tau_{m}) V}. 
	\eeq 
	As in~Ref.~\onlinecite{Frohlich2006} we can rewrite \eqref{temp.32372837} so that
	\beq
		\sp{\delta_{\mathbf{x}}}{ e^{-t(H_{0} + V)} \delta_{\mathbf{x}}} = \int d\nu(\mathbf{n}(\tau)) \exp\left( - \int_{0}^{t} V(\mathbf{n}(\tau)) \, d\tau  \right) \, \prod_{i=1}^m \left[  H_{0}(x_{i}, y_{i}) \sqrt{ \mathbf{n}_{x_{i}}(\tau_{i}+)  \mathbf{n}_{y_{i}}(\tau_{i}-)   }  \right].
	\eeq
	Here, $V(\mathbf{n})$ denotes the potential on the $n$-particle configuration space $\mathbb{Z}^{nd}$ that is induced from the 1-particle potential $\lambda V(u)$, $u \in \mathbb{Z}^d$. Note that only the exponential contains random potentials so that
	\begin{align}
		\EW{\sp{\delta_{\mathbf{x}}}{ e^{-t(H_{0} + V)} \delta_{\mathbf{x}}}} \leq \int d\nu(\mathbf{n}(\tau)) \EW{\exp\left( - \int_{0}^{t} V(\mathbf{n}(\tau)) \, d\tau  \right)} \, \prod_{i=1}^m \left[  | H_{0}(x_{i}, y_{i}) | \sqrt{ \mathbf{n}_{x_{i}}(\tau_{i}+)  \mathbf{n}_{y_{i}}(\tau_{i}-)   }  \right]	
	\end{align}
	($| H_{0} |$ denotes the operator that emerges from $H_{0}$ when replacing all the matrix elements of $H_{0}$ by their absolute value). Observe that
	\beq
		\EW{\exp\left( - \int_{0}^{t} V(\mathbf{n}(\tau)) \, d\tau  \right)}   =   \EW{    e^{ - nt (nt)^{-1}  \sum_{\mathbf{u} \in \Omega} \int_{0}^t \lambda V(u) N_{u}(\mathbf{n}(s)) ds  }   } 
	\eeq
	where 
	\beq
		\frac{1}{n}  \sum_{u \in \Omega} \frac{1}{t} \int_{0}^t (\cdot) N_{u}(\mathbf{n}(s)) ds
	\eeq
	can be regarded as an averaging operation. Pulling it outside of the exponential function and the expectation value by using Jensen's inequality we arrive at
	\begin{align}
		\EW{    e^{ - nt (nt)^{-1}  \sum_{u \in \Omega} \int_{0}^t \lambda V(u) N_{u}(\mathbf{n}(s)) ds  }   }
		&\leq	\frac{1}{n}  \sum_{u \in \Omega} \frac{1}{t} \int_{0}^t       \EW{    e^{ - nt  |\lambda | V(u) }   }    N_{u}(\mathbf{n}(s)) ds \\
		&=	\EW{    e^{ - nt  |\lambda | V(0) }   } \frac{1}{n}  \sum_{u \in \Omega} \frac{1}{t} \int_{0}^t  N_{u}(\mathbf{n}(s)) ds \\
		&=	\EW{    e^{ - nt  |\lambda | V(0) }   }.
	\end{align}
	This yields
	\begin{align}
		\EW{\sp{\delta_{\mathbf{x}}}{ e^{ - t(H_{0} + V)} \delta_{\mathbf{x}}}}   
		&\leq	   \EW{    e^{- nt  |\lambda | V(0) }   }   \int d\nu(\mathbf{n}(\tau))  \prod_{i=1}^m \left[  | H_{0}(x_{i}, y_{i}) | \sqrt{ \mathbf{n}_{x_{i}}(\tau_{i}+)  \mathbf{n}_{y_{i}}(\tau_{i}-)   }  \right] \nn \\
		&=	   \EW{    e^{- nt  |\lambda | V(0) }   }  \sp{\delta_{\mathbf{x}}}{ e^{-t | H_{0} |} \delta_{\mathbf{x}}} \nn \\
		&\leq	   \EW{    e^{- nt  |\lambda | V(0) }   }    e^{-t | \max \mathrm{supp} H_{0} | h_{\max}}
 	\end{align}
	where $\max \mathrm{supp} H_{0} := \bigcup_{\mathbf{k}} \left[ \mathrm{supp} H_0(\mathbf{k},\cdot) - \mathbf{k} \right]$ denotes an upper bound on the $\mathbf{k}$-dependent support of $H_{0}(\mathbf{k}, \cdot)$ and $h_{\max} := \max_{\mathbf{n}, \mathbf{m}} | H(\mathbf{n},\mathbf{m}) |$ and . A similar bound holds for $t > 0$. Putting everything together with $t := - \min\{ 1, (n | \lambda |)^{-1} \}$ we get the desired upper bound.
		
\end{proof}

To prove dynamical localization of the toric code defects we have to show that the sufficient criterium~\eqref{np.aizenman.criterium} holds true for our $n$-body Hamiltonian (see Theorem~2) which specifies interactions between the defects in terms of inhomogeneous, finite-range hopping matrix elements. To that purpose we need to restate and prove Lemma 5.1, Theorem 5.3, Theorem 6.1, Lemma 6.3. and the induction step in Sect.~6 of~Ref.~\onlinecite{AizenmanWarzel2009}. We start with the formulation and the proof of the main theorem. Lemma 5.1, Theorem 5.3, Theorem 6.1 and Lemma 6.3 follow afterwards. The main theorem is formulated for the presence of electric charges and absence of magnetic charges. Its adaption to the presence of magnetic charges and absence of electric charges is immediate. Note that the Hamiltonian $P_{n_{e},0} (H_{I} +\lambda H_{r}) P_{n_{e},0}$ from the main text (cf.~(4) and~(7) in the main text) induces quantum dynamics for the $2n_{e}$ electric charges on the lattice $\mathbb{Z}^d$ ($d=2$ in case of the toric code). We denote the Hamiltonian that describes this evolution of the $2n_{e}$ interacting electric charges on the lattice $\mathbb{Z}^d$ by $H^{(2n_{e})}$.

\begin{theorem}[Dynamical Localization]
	Let $n \in 2 \mathbb{N}$ ($n < \infty$) denote the number of electric charges and let $H^{(n)}$ be the Hamiltonian describing the evolution of the $n$ electric charges on the lattice $\mathbb{Z}^d$ that corresponds to the spin-Hamiltonian $P_{n/2,0} (H_{I} +\lambda H_{r}) P_{n/2,0}$ (cf.~(4) and~(7) in the main text).  Then for each $m \in \mathbb{N}$ (cf.~(4) in the main text) there is a $\lambda_{0} \in (0,\infty)$ with the property that for all $\lambda \geq \lambda_{0}$ there are positive constants $A,\xi < \infty$ such that 
	\beq
		\mathbb{E}\left[ \sup_{t \in \AR} \large|  \langle  \delta_{\mathbf{x}} , e^{-it H^{(n)}}  \delta_{\mathbf{y}}      \rangle   \large|  \right] \leq A \, e^{-\text{dist}_{\mathcal{H}}(\mathbf{x}, \mathbf{y}) / \xi  }.
	\eeq
	for all $\mathbf{x}, \mathbf{y} \in (\mathbb{Z}^{2})^2$. Here $\mathbb{E}[...]$ denotes the expectation value with respect to the random potential variables $V(x)$, $x \in \mathbb{Z}^d$.
\end{theorem}

\begin{proof}
	The theorem is proven inductively in the particle number $n$; the validity of the fractional moment criterium~\eqref{np.aizenman.criterium} in the 1-particle case ($n=1$, the induction anchor) has been proven in~Ref.~\onlinecite{Aizenman1993}. Consequently, it is left to show that the validity of the fractional moment criterium~\eqref{np.aizenman.criterium} for $n-1$ particles implies the validity of the fractional moment criterium for $n$ particles. Without loss of generality we can assume that $\Omega \subset \mathbb{Z}^d$ is chosen such that $\mathbf{x}, \mathbf{y} \in C^{(n)}(\Omega)$ (otherwise, the LHS of~\eqref{np.aizenman.criterium} vanishes; recall the definition of $H_{\Omega}$ from Theorem~\ref{np.FMM}). Choose $L_{0} \in \mathbb{N}$ arbitrarily (we will fix the value of $L_{0}$ at the end of the proof) and set $L_{k+1} := 2(L_{k} +1)$ as in~Ref.~\onlinecite{AizenmanWarzel2009}. Let $x, y \in \mathbb{Z}^d$ be those vectors that realize the Hausdorff pseudo distance, i.e.,  $\distH{x}{y} = | x-y |$, and assume without loss of generality that $| x-y | > L_{0}$ . Consequently, there exists a unique $k \in \mathbb{N}_{0}$ such that $y \notin \Lambda_{L_{k}}(x)$ but $y \in \Lambda_{L_{k+1}}(x)$. Assume that $x$ and $y$ are sufficiently far apart to guarantee that $L_{k} \geq 4m$. Note that there exists $c < \infty$ such that $L_{k} \leq \distH{x}{y} = | x-y | \leq c L_{k}$. The remainder of the proof is divided into two parts. In the first part we assume $\mathrm{diam}(\mathbf{x}) \geq \frac{L_{k}}{2}$. In the second part of the proof we will thus have to deal with $n$-defect configurations $\mathbf{x} \in \mathbb{Z}^{nd}$ with the property $\mathrm{diam}(\mathbf{x}) < \frac{L_{k}}{2}$. In case of case $\mathrm{diam}(\mathbf{x}) \geq \frac{L_{k}}{2}$, Eq.~(A.2) in~Ref.~\onlinecite{AizenmanWarzel2009}, i.e.
	\beq
		l(\mathbf{x}) \geq \frac{1}{n-1} \text{diam}(\mathbf{x}),
	\eeq
	yields the lower bound
	\beq
		l(\mathbf{x}) \geq \frac{1}{n-1} \text{diam}(\mathbf{x}) \geq \frac{L_{k}}{2(n-1)} \geq \frac{\distH{x}{y}}{2c (n-1)}.
	\eeq
	for $l(\mathbf{x})$. We conclude that in case of $\text{diam}(\mathbf{x}) \geq \frac{L_{k}}{2}$ the Theorem is a direct consequence of Theorem~\ref{theorem.5.3} because for instance in case of $l(\mathbf{x}) \geq l(\mathbf{y})$ and $\distH{x}{y} > l(\mathbf{x})$ we get 
	\beq
		\sup_{\substack{ I \subset \AR \\ | I | \geq 1  }} \sup_{ \Omega \subseteq \mathbb{Z}^d }  \avEW{  |  G_{\Omega}^{(n)}(\mathbf{x},\mathbf{y})  |^s      }     \leq A \, e^{-  \frac{1}{\xi} l(\mathbf{x}) } \leq  A \, e^{-  \frac{1}{\xi} \frac{\distH{x}{y}}{2c (n-1)} }.
	\eeq
	This observation implies dynamical localization (see Theorem~\ref{np.FMM}). We will thus assume $\text{diam}(\mathbf{x}) < \frac{L_{k}}{2}$ in the remainder of this proof. Hence, $\mathbf{x} \in C^{(n)}(\Omega \cap \Lambda_{L_{k}}(x))$. With the resolvent identity we can remove all the terms in $H_{\Omega}$ which connect $C^{(n)}(\Omega \cap \Lambda_{L_{k}}(x))$ to its complement in $C^{(n)}(\Omega)$, i.e., $C^{(n)}(\Omega,\Omega \backslash \Lambda_{L_{k}}(x))$:
	\begin{align}
		|G_{\Omega}(\mathbf{x},\mathbf{y}, z)|	=& |\sp{\dirac{x}}{(H_{\Omega} - z)^{-1} \dirac{y} }| \nn \\
										=& |\sp{\dirac{x}}{(H_{C^{(n)}(\Omega \cap \Lambda_{L_{k}}(x))} \oplus H_{C^{(n)}(\Omega,\Omega \backslash \Lambda_{L_{k}}(x))}  -z   )^{-1} \Gamma_{C^{(n)}(\Omega,\Omega \backslash \Lambda_{L_{k}}(x))}^{C^{(n)}(\Omega \cap \Lambda_{L_{k}}(x))}  (H_{\Omega} - z)^{-1} \dirac{y} } | \nn \\
										\leq& \sum_{\mathbf{w}, \mathbf{w'} \in \overline{\partial^{(m)} \Lambda_{L_{k}}(x) }}   |\sp{\dirac{x}}{(H_{C^{(n)}(\Omega \cap \Lambda_{L_{k}}(x))} \oplus H_{C^{(n)}(\Omega,\Omega \backslash \Lambda_{L_{k}}(x))}  -z   )^{-1}  \dirac{w} }   \nn \\
										&   \times  \sp{ \dirac{w}}{ \Gamma_{C^{(n)}(\Omega,\Omega \backslash \Lambda_{L_{k}}(x))}^{C^{(n)}(\Omega \cap \Lambda_{L_{k}}(x))} \dirac{w'}} \sp{\dirac{w'}} {(H_{\Omega} - z)^{-1} \dirac{y} } |
	\end{align}
	with
	\begin{align}\label{pf.Lemma.6.3.global.1asdsd}
		\overline{\partial^{(m)} M } := & \left. \left\{   \mathbf{r} = (r_{1},...,r_{n}) \in C^{(n)}(M)  \right| \exists \mathbf{s} \in  \mathbb{Z}^{(nd)} \backslash C^{(n)}(M) \text{ such that }  \max\{ \| s_{i} - r_{j}   \|_{1} \leq m \, | \, i,j \in \{ 1,...,n \} \}  \right\} \nn \\
								& \cup \left.\left\{   \mathbf{r} = (r_{1},...,r_{n}) \in \mathbb{Z}^{(nd)} \backslash C^{(n)}(M) \right|  \exists \mathbf{s} \in  C^{(n)}(M) \text{ such that }  \max\{ \| s_{i} - r_{j}   \|_{1} \leq m \, | \, i,j \in \{ 1,...,n \} \}   \right\}.
	\end{align}
	Here $\Gamma_{C^{(n)}(\Omega,\Omega \backslash \Lambda_{L_{k}}(x))}^{C^{(n)}(\Omega \cap \Lambda_{L_{k}}(x))}$ denotes the ``boundary strip operator'' $H_{\Omega} - H_{C^{(n)}(\Omega \cap \Lambda_{L_{k}}(x))} \oplus H_{C^{(n)}(\Omega,\Omega \backslash \Lambda_{L_{k}}(x))}$ and we have used 
	\beq
		\left\{  (\mathbf{q}, \mathbf{q'}) \in \mathbb{Z}^{nd} \times \mathbb{Z}^{nd}  | \sp{\dirac{q}}{\Gamma_{C^{(n)}(\Omega,\Omega \backslash \Lambda_{L_{k}}(x))}^{C^{(n)}(\Omega \cap \Lambda_{L_{k}}(x))} \dirac{q'}  }  \neq 0 \right\} \subseteq \overline{\partial^{(m)} \Lambda_{L_{k}}(x) } \times \overline{\partial^{(m)} \Lambda_{L_{k}}(x) }.
	\eeq
	Together with $| \sp{\dirac{q}}{\Gamma_{C^{(n)}(\Omega,\Omega \backslash \Lambda_{L_{k}}(x))}^{C^{(n)}(\Omega \cap \Lambda_{L_{k}}(x))} \dirac{q'}  } | \leq h_{\max}$ for all $\mathbf{q}, \mathbf{q'} \in \mathbb{Z}^{nd}$ we thus arrive at
	\beq
		|G_{\Omega}(\mathbf{x},\mathbf{y}, z)|	\leq 	h_{\max}  \sum_{\mathbf{w}, \mathbf{w'} \in \overline{\partial^{(m)} \Lambda_{L_{k}}(x) }}   |\sp{\dirac{x}}{(H_{C^{(n)}(\Omega \cap \Lambda_{L_{k}}(x))} \oplus H_{C^{(n)}(\Omega,\Omega \backslash \Lambda_{L_{k}}(x))}  -z   )^{-1}  \dirac{w} } \sp{\dirac{w'}} {(H_{\Omega} - z)^{-1} \dirac{y} } |
	\eeq
	and therefore (recall that $\mathbf{x} \in C^{(n)}(\Omega \cap \Lambda_{L_{k}}(x))$)
	\begin{align}
		|G_{\Omega}(\mathbf{x},\mathbf{y}, z)|	\leq &	h_{\max}  \sum_{\substack{\mathbf{w} \in \overline{\partial^{(m)} \Lambda_{L_{k}}(x) } \cap C^{(n)}(\Lambda_{L_{k}}(x))  \\ \mathbf{w'} \in \overline{\partial^{(m)} \Lambda_{L_{k}}(x) }    }}   |\sp{\dirac{x}}{(H_{C^{(n)}(\Omega \cap \Lambda_{L_{k}}(x))}  -z   )^{-1}  \dirac{w} } \sp{\dirac{w'}} {(H_{\Omega} - z)^{-1} \dirac{y} } | \nn \\
										= &		h_{\max}  \sum_{\substack{\mathbf{w} \in \overline{\partial^{(m)} \Lambda_{L_{k}}(x) } \cap C^{(n)}(\Lambda_{L_{k}}(x))  \\ \mathbf{w'} \in \overline{\partial^{(m)} \Lambda_{L_{k}}(x) }    }}  | G_{\Omega \cap \Lambda_{L_{k}}(x)}(\mathbf{x}, \mathbf{w}) | \, | G_{\Omega}(\mathbf{w'}, \mathbf{y}) |.
	\end{align}
	Hence, Theorem 2.1 (the Wegner-type estimate) of Aizenman and Warzel's paper~\cite{AizenmanWarzel2009} gives (using $| a+b |^s \leq |a|^s + |b|^s$, $s \in (0,1)$)
	\beq
		\avEW{|G_{\Omega}(\mathbf{x},\mathbf{y})|^s}	\leq		h_{\max}^s \frac{C}{| \lambda |^s}  \,  \left| \overline{\partial^{(m)} \Lambda_{L_{k}}(x) } \right| \sum_{\mathbf{w} \in \overline{\partial^{(m)} \Lambda_{L_{k}}(x) } \cap C^{(n)}(\Lambda_{L_{k}}(x)) }   \avEW{| G_{\Omega \cap \Lambda_{L_{k}}(x)}(\mathbf{x}, \mathbf{w}) |^s }. 
	\eeq
	Note that $\left| \overline{\partial^{(m)} \Lambda_{L_{k}}(x) } \right|$ grows only polynomially in $L_{k}$. Therefore, this factor does not prevent exponential bounds in terms of $L_{k}$. In analogy to~\eqref{pf.Lemma.6.3.global.sum.rewriting} we can rewrite the sum to get (see also~\eqref{def.wide.inner.bdry})
	\begin{align}
		\avEW{|G_{\Omega}(\mathbf{x},\mathbf{y})|^s}	\leq &	h_{\max}^s \frac{C}{| \lambda |^s} \, \left| \overline{\partial^{(m)} \Lambda_{L_{k}}(x) } \right|   \sum_{u \in \partial^{(m)} \Lambda_{L_{k}}(x)} \sum_{\mathbf{w} \in C^{(n)}(\Lambda_{L_{k}}(x) \cap \Omega;u)}  \avEW{| G_{\Omega \cap \Lambda_{L_{k}}(x)}(\mathbf{x}, \mathbf{w}) |^s } \nn \\
												\leq &      h_{\max}^s \frac{C}{| \lambda |^s} \, \left| \overline{\partial^{(m)} \Lambda_{L_{k}}(x) } \right|  \sup_{\tilde{\Omega} : \tilde{\Omega} \subseteq \Lambda_{L_{k}}(x) }  \sum_{u \in \partial^{(m)} \Lambda_{L_{k}}(x)} \sum_{\mathbf{w} \in C^{(n)}(\tilde{\Omega};u)}  \avEW{| G_{\tilde{\Omega}}(\mathbf{x}, \mathbf{w}) |^s }.
	\end{align}
	Consequently,
	\beq \label{pf.main.thm.temp.100}
		\sup_{\substack{   I \subset \AR \\ | I | \geq 1    }} \avEW{|G_{\Omega}(\mathbf{x},\mathbf{y})|^s}  \leq   h_{\max}^s \frac{C}{| \lambda |^s} \, \left| \overline{\partial^{(m)} \Lambda_{L_{k}}(x) } \right| \left( S_{1} + S_{2} \right)
	\eeq
	with
	\begin{align}
		S_{1}	:=&	  \sup_{\substack{   I \subset \AR \\ | I | \geq 1    }}    \sup_{\tilde{\Omega} : \tilde{\Omega} \subseteq \Lambda_{L_{k}}(x) }  \sum_{u \in \partial^{(m)} \Lambda_{L_{k}}(x)} \sum_{\mathbf{w} \in C^{(n)}(\tilde{\Omega};u) \backslash C^{(n)}_{r_{L_{k}}}(\tilde{\Omega};u)}  \avEW{| G_{\tilde{\Omega}}(\mathbf{x}, \mathbf{w}) |^s },   \nn \\
		S_{2}	:=& 	  \sup_{\substack{   I \subset \AR \\ | I | \geq 1    }}    \sup_{\tilde{\Omega} : \tilde{\Omega} \subseteq \Lambda_{L_{k}}(x) }  \sum_{u \in \partial^{(m)} \Lambda_{L_{k}}(x)} \sum_{\mathbf{w} \in C^{(n)}_{r_{L_{k}}}(\tilde{\Omega};u)}  \avEW{| G_{\tilde{\Omega}}(\mathbf{x}, \mathbf{w}) |^s }.
	\end{align}
	We continue with the estimation of the term $S_{1}$: By Theorem~\ref{theorem.5.3}
	\begin{align}
		 \sum_{u \in \partial^{(m)} \Lambda_{L_{k}}(x)} \sum_{\mathbf{w} \in C^{(n)}(\tilde{\Omega};u) \backslash C^{(n)}_{r_{L_{k}}}(\tilde{\Omega};u)}  \avEW{| G_{\Omega \cap \Lambda_{L_{k}}(x)}(\mathbf{x}, \mathbf{w}) |^s }  \nn \\
		\leq	 \sum_{u \in \partial^{(m)} \Lambda_{L_{k}}(x)} \sum_{\mathbf{w} \in C^{(n)}(\tilde{\Omega};u) \backslash C^{(n)}_{r_{L_{k}}}(\tilde{\Omega};u)}   A \, e^{- \frac{1}{\xi}  \min\{  \distH{\mathbf{x}}{\mathbf{w}} ,  \max\{ l(\mathbf{x}), l(\mathbf{w})  \}   \}    } \label{pf.main.thm.temp.1}
	\end{align}
	with
	\begin{align}
		 \distH{x}{w}	:=&		\max \left\{    \max_{1\leq i \leq k} \text{dist}(x_{i}, \mathbf{w} ) ,  \max_{1\leq i \leq k} \text{dist}(\mathbf{x}, w_{i} )    \right\} \nn \\
		 			   \geq&	    \text{dist}(\mathbf{x}, u) \nn \\
					   \geq&	    \frac{L_{k}}{2} - m \nn \\
					   \geq&	    \frac{1}{n-1}  \left( \frac{L_{k}}{2}  - m  \right) \label{pf.main.thm.temp.2}
	\end{align}
	because $u \in \partial^{(m)} \Lambda_{L_{k}}(x)$ and $\text{diam}(\mathbf{x}) < \frac{L_{k}}{2}$ by assumption. On the other hand  $\mathbf{w} \in C^{(n)}(\tilde{\Omega};u) \backslash C^{(n)}_{r_{L_{k}}}(\tilde{\Omega};u)$ implies $\text{diam}(\mathbf{w}) \geq \frac{L_{k}}{2}$ and therefore
	\beq   \label{pf.main.thm.temp.3}
		l(\mathbf{w}) \geq \frac{1}{n-1} \frac{L_{k}}{2} \geq \frac{1}{n-1}  \left( \frac{L_{k}}{2}  - m  \right). 
	\eeq
	The use of~\eqref{pf.main.thm.temp.2} and ~\eqref{pf.main.thm.temp.3} in~\eqref{pf.main.thm.temp.1} yields
	\begin{align}
		 \sum_{u \in \partial^{(m)} \Lambda_{L_{k}}(x)} \sum_{\mathbf{w} \in C^{(n)}(\tilde{\Omega};u) \backslash C^{(n)}_{r_{L_{k}}}(\tilde{\Omega};u)}  \avEW{| G_{\Omega \cap \Lambda_{L_{k}}(x)}(\mathbf{x}, \mathbf{w}) |^s }  \leq&	 \sum_{u \in \partial^{(m)} \Lambda_{L_{k}}(x)} \sum_{\mathbf{w} \in C^{(n)}(\tilde{\Omega};u) \backslash C^{(n)}_{r_{L_{k}}}(\tilde{\Omega};u)}   A \, e^{- \frac{1}{\xi}  \frac{1}{n-1}  \left( \frac{L_{k}}{2}  - m  \right)   }  \nn \\
		 \leq&	   A \, e^{- \frac{1}{\xi}  \frac{1}{n-1}  \left( \frac{L_{k}}{2}  - m  \right)   } \,   | \partial^{(m)} \Lambda_{L_{k}}(x) | \,  |C^{(n)}(\Lambda_{L_{k}})| \nn \\
		 \leq&	   \tilde{A} e^{-\frac{L_{k}}{\tilde{\xi}}}
	\end{align}
	for appropriate definitions of $\tilde{A}$ and $\tilde{\xi}$ because the quantities $| \partial^{(m)} \Lambda_{L_{k}}(x) |$ and $| C^{(n)}(\Lambda_{L_{k}})|$ grow only polynomially in $L_{k}$. The assumptions $\distH{x}{y} = |x-y|$, $y \notin \Lambda_{L_{k}}(x)$ but $y \in \Lambda_{L_{k+1}}(x)$ with $L_{k+1} := 2(L_{k} +1)$ at the beginning of the proof imply $|x - y| \leq 2 ( L_{k} + 1)$ and thus $\distH{x}{y} \leq 2 ( L_{k} + 1)$. We conclude that after the redefinitions of $\tilde{\xi}$ and $\tilde{A}$ we arrive at
	\beq\label{pf.main.thm.ub.for.S1}
		S_{1} \leq	\tilde{\tilde{A}} e^{-\frac{\distH{x}{y}}{\tilde{\tilde{\xi}}}}.
	\eeq
	To finish the proof of the theorem we still have to upper bound the summand $S_{2}$: we get
	\begin{align}\label{pf.main.thm.ub.for.S2}
		 S_{2}	=&		\sup_{\substack{   I \subset \AR \\ | I | \geq 1    }}  \sup_{\tilde{\Omega} \subseteq \Lambda_{L_{k}}(0) }  \sum_{u \in \partial^{(m)} \Lambda_{L_{k}}(0)} \sum_{\mathbf{w} \in C^{(n)}_{r_{L_{k}}}(\tilde{\Omega};u)}  \avEW{| G_{\tilde{\Omega}}(\mathbf{x}, \mathbf{w}) |^s }  \nn \\
		 \leq&		\sup_{\substack{I \subseteq \AR \\ | I | \geq 1}} \sup_{\tilde{\Omega} \subseteq \Lambda_{L_{k}}(0)}  | \partial^{(m)} \Lambda_{L_{k}} |   \sum_{u \in \partial^{(m)} \Lambda_{L_{k}}(0)} \sum_{\substack{ \mathbf{x} \in C^{(n)}_{r_{L_{k}-4m}}(\tilde{\Omega};0)  \\  \mathbf{y} \in C^{(n)}_{r_{L_{k}}}(\tilde{\Omega};u) }}    \avEW{  |  G^{(n)}_{\tilde{\Omega}}(\mathbf{x},\mathbf{y})  |^s} \nn \\
		 =&  B_{s}^{(n)}(L_{k})
	\end{align}
	(see~\eqref{def.B}) using the translation invariance of the expectation values for the first equality. The second inequality is simply a consequence of adding a factor $\geq 1$ and new summands to the sum. Therefore (use \eqref{pf.main.thm.ub.for.S1} and \eqref{pf.main.thm.ub.for.S2} in \eqref{pf.main.thm.temp.100}),
	\beq\label{temp.sjfijfiejf}
		\avEW{|G_{\Omega}(\mathbf{x},\mathbf{y})|^s}	\leq 	h_{\max}^s \frac{C}{| \lambda |^s} \left( \tilde{\tilde{A}} e^{-\frac{\distH{x}{y}}{\tilde{\tilde{\xi}}}} + B_{s}^{(n)}(L_{k}) \right).
	\eeq
	It is thus left to show that $B_{s}^{(n)}(L_{k})$ decays exponentially in the distance $\distH{x}{y}$. Because of the assumed uniform $(n-1)$ particle localization Theorem~\ref{theorem.6.1} implies that there exist $a, A, p < \infty$ and $\nu > 0$ (these quantities depend on the localization properties of the $n-1$ particle system) such that
	\beq
		B^{(n)}_{s}(L_{k+1}) \leq \frac{a}{\lambda^s} B^{(n)}_{s}(L_{k})^2 + A L_{k+1}^{2p} e^{-2 \nu L_{k}}.
	\eeq 
	To prove the exponential decay of $B^{(n)}_{s}(L_{k})$ we need the Lemma 6.2 of Aizenman and Warzel's paper~\cite{AizenmanWarzel2009}:
	
	\begin{lemma}[Lemma 6.2]\label{Lemma.6.2}
		Assume that $S: \, \AR \rightarrow \AR$, $q, b, p, \eta \in [ 0, \infty )$ and $\nu, \schweif{L}_{0} \in (0, \infty)$ satisfy
		\beqa
			S(2^k \schweif{L}_{0} )	
			&\leq&	q S( 2^{k-1} \schweif{L}_{0}  )^2 + b ( 2^{k-1} \schweif{L}_{0}  )^{2p} e^{-2\nu (2^{k-1} \schweif{L}_{0} )}, \label{Lemma.6.2.Condition.1} \\
			\eta^2
			&\geq&	q b + \eta \frac{2^p}{\schweif{L}_{0}^p}, \label{Lemma.6.2.Condition.2} \\
			1
			&>&		q S(\schweif{L}_{0} ) + \eta \schweif{L}_{0}^p e^{- \eta \schweif{L}_{0} } =: e^{-\mu \schweif{L}_{0} }  \label{Lemma.6.2.Condition.3}
		\eeqa
		for all $k \in \mathbb{N}_{0}$. Then,
		\beq
			S(2^k \schweif{L}_{0} ) \leq \frac{1}{q} e^{-\mu 2^k \schweif{L}_{0}} 
		\eeq
		for all $k \in \mathbb{N}_{0}$.
	\end{lemma}
	
	\begin{proof}
		Define 
		\[
			R(L) := q S(L) + \eta L^p e^{- \nu L}
		\]
		and observe that
		\begin{align}
			R(2^k \schweif{L}_{0})
			&=		q S(2^k \schweif{L}_{0}) + \eta (2^k \schweif{L}_{0})^p e^{- \nu 2^k \schweif{L}_{0}} \nn \\
			&\leq	q \left(   q S(2^{k-1} \schweif{L}_{0})^2 + b (2^{k-1} \schweif{L}_{0})^{2p} e^{-2\nu (2^{k-1} \schweif{L}_{0})}   \right)   +  \eta 2^p (2^{k-1} \schweif{L}_{0})^p e^{- \nu 2^k \schweif{L}_{0}} \nn \\
			&=		\left( q S(2^{k-1} \schweif{L}_{0}) \right)^2 + \left(  qb  +  \eta \frac{2^p}{(2^{k-1} \schweif{L}_{0})^p}  \right) (2^{k-1} \schweif{L}_{0})^{2p} \left(  e^{-\nu 2^{k-1} \schweif{L}_{0}}  \right)^2 \nn \\
			&\leq	\left( q S(2^{k-1} \schweif{L}_{0}) \right)^2 + \eta^2 (2^{k-1} \schweif{L}_{0})^{2p} \left(  e^{-\nu 2^{k-1} \schweif{L}_{0}}  \right)^2 \nn \\
			&\leq	\left( q S(2^{k-1} \schweif{L}_{0}) \right)^2 + \left(    \eta (2^{k-1} \schweif{L}_{0})^{p}   e^{-\nu 2^{k-1} \schweif{L}_{0}}       \right)^2 \nn   \\
			&\leq	\left( q S(2^{k-1} \schweif{L}_{0})   +   \eta (2^{k-1} \schweif{L}_{0})^{p}   e^{-\nu 2^{k-1} \schweif{L}_{0}}     \right)^2 \nn \\
			&=		R(2^{k-1} \schweif{L}_{0})^2.
		\end{align}
		Consequently,
		\begin{align}
			S(2^k \schweif{L}_{0})	
			&\leq	\frac{1}{q}  R(2^k \schweif{L}_{0})  \leq  \frac{1}{q}  R(2^{k-1} \schweif{L}_{0})^2 \leq ... \nn \\
			&\leq	\frac{1}{q}  R( \schweif{L}_{0})^{2^k}  \nn \\
			&=		\frac{1}{q} e^{- \mu 2^k \schweif{L}_{0}}
		\end{align}
		for $R(\schweif{L}_{0}) = e^{-\mu \schweif{L}_{0}}$.
	\end{proof}
	
	We proceed by using Lemma~\ref{Lemma.6.2} to find an exponential upper bound for $B^{(n)}_{s}(L_{k})$ ($S(\tilde{L}_{k}) := B^{(n)}_{s}(L_{k}), \tilde{L}_{k} := 2^k(L_{0} + 2)$ as in~Ref.~\onlinecite{AizenmanWarzel2009}) with respect to $L_{k}$. If we define $\schweif{L}_{0} := L_{0} + 2$ we get
	\beq\label{temp1234522}
		B^{(n)}_{s}(L_{k}) = S(\tilde{L}_{k}) = S(2^k (L_{0} + 2)) = S(2^k \schweif{L}_{0}).
	\eeq
	Recall that we have not fixed the value of $L_{0}$ so far. Hence, our next goal is to fix $L_{0}$ such that $\schweif{L}_{0} = L_{0} + 2$ and $S(\cdot)$ satisfy the conditions~\eqref{Lemma.6.2.Condition.1}, \eqref{Lemma.6.2.Condition.2} and~\eqref{Lemma.6.2.Condition.3} in Lemma~\ref{Lemma.6.2}. We start with the condition~\eqref{Lemma.6.2.Condition.1}:
	\begin{align}
		S(2^k \schweif{L}_{0})	
		&=	S(2 \cdot 2^{k-1}(L_{0} + 2)) = S(2 \tilde{L}_{k-1}) \nn \\
		&= S(\tilde{L}_{k}) = B^{(n)}_{s}(L_{k}) \nn \\
		&\leq \frac{a}{\lambda^s} B^{(n)}_{s}(L_{k-1})^2 + A L_{k}^{2p} e^{-2 \nu L_{k-1}} \label{ineq.that.uses.thm.6.1} \\
		&= \frac{a}{\lambda^s} S(\tilde{L}_{k-1})^2 + \tilde{A} \tilde{L}_{k-1}^{2p} e^{-2 \nu \tilde{L}_{k-1}} \nn \\
		&=	\frac{a}{\lambda^s} S(2^{k-1} (L_{0} + 2))^2  + \tilde{A}  ( 2^{k-1} (L_{0} + 2) )^{2p}  e^{-2 \nu 2^{k-1} (L_{0} + 2)} \nn \\
		&=	\frac{a}{\lambda^s} S(2^{k-1}  \schweif{L}_{0})^2  + \tilde{A}  ( 2^{k-1}  \schweif{L}_{0} )^{2p}  e^{-2 \nu 2^{k-1}  \schweif{L}_{0}}
	\end{align}
	(note that $L_{k} = \tilde{L}_{k} - 2$ for all $k$). In inequality~\eqref{ineq.that.uses.thm.6.1} we have used Theorem~\ref{theorem.6.1}. Consequently, condition~\eqref{Lemma.6.2.Condition.1} is fulfilled independently of the choice of $L_{0}$. To satisfy the conditions \eqref{Lemma.6.2.Condition.2} and~\eqref{Lemma.6.2.Condition.3} we choose $\schweif{L}_{0}$, such that 
	\beq\label{Def.strong.disorder.a}
		2^{p+1} e^{-\nu \schweif{L}_{0}} < \frac{1}{2}
	\eeq
	(note that a choice of $\schweif{L}_{0}$ also fixes $L_{0}$ because $\schweif{L}_{0} = L_{0} + 2$ by definition) and $\eta$ such that
	\beq\label{eta.regime}
		\frac{2^{p+1}}{\schweif{L}_{0}^p} < \eta < \frac{1}{2 \schweif{L}_{0}^p} e^{\nu \schweif{L}_{0}}
 	\eeq
	(the existence of this regime is guaranteed by \eqref{Def.strong.disorder.a}). The left inequality leads to the observation that
	\beq\label{temp12345fjiejfi}
		\frac{a}{\lambda^s} \tilde{A}    \leq    \frac{\eta^2}{2} 
	\eeq
	is sufficient for the assumption~\eqref{Lemma.6.2.Condition.2} in Lemma~\ref{Lemma.6.2} (using the correspondences $\frac{a}{\lambda^s} \leftrightarrow q$ and $\tilde{A} \leftrightarrow b$) because~\eqref{temp12345fjiejfi} plus $\frac{\eta}{2}$ times the left inequality of~\eqref{eta.regime} gives
	\beq
		\frac{a}{\lambda^s} \tilde{A} + \frac{\eta}{2} \frac{2^{p+1}}{\schweif{L}_{0}^p}    \leq    \frac{\eta^2}{2}  +   \frac{\eta^2}{2}.    
	\eeq
	We note that the condition \eqref{temp12345fjiejfi} can be satisfied by choosing $\lambda$ large enough. The assumption~\eqref{Lemma.6.2.Condition.3} of Lemma~\ref{Lemma.6.2} reads
	\beq
		\frac{a}{\lambda^s} S(\schweif{L}_{0}) + \eta \schweif{L}_{0}^p e^{-\nu \schweif{L}_{0}} < 1.
	\eeq
	According to the right inequality in \eqref{eta.regime} we can rewrite the LHS of this condition,
	\begin{align}
		\frac{a}{\lambda^s} S(\schweif{L}_{0}) + \eta \schweif{L}_{0}^p e^{-\nu \schweif{L}_{0}}	&= \frac{a}{\lambda^s} B^{(n)}_{s}(\schweif{L}_{0} - 2) + \eta \schweif{L}_{0}^p e^{-\nu \schweif{L}_{0}} \\
			&< \frac{a}{\lambda^s} B^{(n)}_{s}(\schweif{L}_{0} - 2) + \frac{1}{2},
	\end{align}
	and arrive at the sufficient condition
	\beq\label{condition.2352}
		\frac{a}{\lambda^s} B^{(n)}_{s}(\schweif{L}_{0} - 2)    < \frac{1}{2} .
	\eeq
	This demand can again be satisfied by choosing $\lambda$ large enough because the quantity $B^{(n)}_{s}(\schweif{L}_{0} - 2)$ can be bounded by
	\beq
		B^{(n)}_{s}(\schweif{L}_{0} - 2) \leq \frac{C n^2}{\lambda^s} (\schweif{L}_{0} - 2)^{2d(n-1)}
	\eeq
	(this is a consequence of the Wegner estimate; cf. (2.2) in~Ref.~\onlinecite{AizenmanWarzel2009}) where $C$ is $\lambda$-independent. This concludes the verification of the assumptions of Lemma~\ref{Lemma.6.2}. Its application  yields
	\beq
		S(2^k \schweif{L}_{0}) \leq \frac{\lambda^s}{a} e^{-\mu 2^k \schweif{L}_{0}}.
	\eeq
	Therefore,
	\begin{align}
		B^{(n)}_{s}(L_{k}) 
		&=		S(\tilde{L}_{k}) = S(2^k (L_{0} + 2)) = S(2^k \schweif{L}_{0}) \nn \\
		&\leq 	\frac{\lambda^s}{a} e^{-\mu 2^k \schweif{L}_{0}}  =      \frac{\lambda^s}{a} e^{-\mu \tilde{L}_{k} } = \frac{\lambda^s}{a} e^{-\mu ( L_{k} + 2)}.
	\end{align}
	Recall that 
	\beq
		\frac{\distH{x}{y}}{c} \leq L_{k}
	\eeq
	because of our observation that $L_{k} \leq \distH{x}{y} = | x-y | \leq c L_{k}$ ($c < \infty$) at the beginning of the proof. We thus arrive at the desired exponential decay
	\beq
		B^{(n)}_{s}(L_{k}) \leq \left( \frac{\lambda^s}{a} e^{- 2 \mu} \right) e^{-\frac{\mu}{c} \distH{x}{y}}.
	\eeq
	We have thus shown the validity of the fractional moment criterium~\eqref{np.aizenman.criterium} for $n$ defects (cf. \eqref{temp.sjfijfiejf}) based on the validity of the fractional moment criterium for $n-1$ particles (this was the induction assumption).  This concludes the inductive proof of the theorem.
\end{proof}

\begin{theorem}[Theorem 5.3]\label{theorem.5.3}
	Assume a system with $n-1$ particles ($n \geq 2$) satisfies the fractional moment criterium~\eqref{np.aizenman.criterium} for all $\lambda \in (\lambda_{0}, \infty)$. Then there exist some $s \in (0,1)$, $A, \xi < \infty$ such that 
	\beq
		\sup_{\substack{ I \subset \AR \\ | I | \geq 1  }} \sup_{ \Omega \subseteq \mathbb{Z}^d }  \avEW{  |  G_{\Omega}^{(n)}(\mathbf{x},\mathbf{y})  |^s      }   \leq    A \, e^{- \frac{1}{\xi}  \min\{  \distH{\mathbf{x}}{\mathbf{y}} ,  \max\{ l(\mathbf{x}), l(\mathbf{y})  \}   \}    }
	\eeq
	for all $\lambda \in (\lambda_{0}, \infty)$ and all $\mathbf{x}, \mathbf{y} \in \mathbb{Z}^{dn}$.
\end{theorem}

\begin{proof}
	As in Ref.~\onlinecite{AizenmanWarzel2009} we assume without loss of generality that $l(\mathbf{x}) \geq l(\mathbf{y})$ and set $J, K \subset \{  1,...,n \}$, $J \dot{\cup} K = \{ 1,...,n \}$ such that
	\beq
		l(\mathbf{x}) = \min_{j \in J, k \in K} | x_{j} - x_{k} |.
	\eeq
	Corresponding to the partition $\{ 1,...,n \} = J \dot{\cup} K$ there is a division of the total $n$-particle system into two subsystems. We set $H_{\Omega}^{(J,K)} = H_{\Omega}^{(J)} \otimes 1 + 1 \otimes H_{\Omega}^{(K)}$ where $H_{\Omega}^{(J)}$ and $H_{\Omega}^{(K)}$ are chosen such that the dynamics of $H_{\Omega}^{(J,K)}$ and $H_{\Omega}$ agree within the subsystems. Thus $H_{\Omega}^{(J,K)}$ emerges from $H_{\Omega}$ by removing all the interactions between the subsystems associated to $J$ and $K$. With the use of the triangle inequality $|a+b|^s \leq |a|^s + |b|^s$ Aizenman and Warzel split up the expression we want to bound:
	\beq \label{temp1.in.pf.of.thm5.3}
		 \avEW{  |  G_{\Omega}(\mathbf{x},\mathbf{y})  |^s      }   \leq  \avEW{  |  G_{\Omega}^{(J,K)}(\mathbf{x},\mathbf{y})  |^s      }  + \avEW{  |  G_{\Omega}^{(J,K)}(\mathbf{x},\mathbf{y})  -  G_{\Omega}(\mathbf{x},\mathbf{y})  |^s      } .
	\eeq
	With Theorem 5.2 in~Ref.~\onlinecite{AizenmanWarzel2009} we get the desired bound for $\avEW{  |  G_{\Omega}^{(J,K)}(\mathbf{x},\mathbf{y})  |^s      }$ using $\distHJK{\mathbf{x}}{\mathbf{y}} \geq \distH{\mathbf{x}}{\mathbf{y}}$:
	\beq \label{temp2.in.pf.of.thm5.3}
		\avEW{  |  G_{\Omega}^{(J,K)}(\mathbf{x},\mathbf{y})  |^s      }  \leq  A e^{- \distH{\mathbf{x}}{\mathbf{y}} / \xi}
	\eeq
	Thus, we are left with the task to bound $\avEW{  |  \Delta  |^s      }$ for
	\beq
		\Delta := G_{\Omega}^{(J,K)}(\mathbf{x},\mathbf{y},z)  -  G_{\Omega}(\mathbf{x},\mathbf{y},z).
	\eeq
	The next step in the original paper is the application of the H\"older inequality:
	\begin{align}
		\avEW{  |  \Delta  |^s }	&\leq	\left(    \avEW{  |  \Delta  |^{\frac{3s(1+s)}{2(1+2s)} } }      \right)^{\frac{1+2s}{2+s}}  \left(    \avEW{  |  \Delta  |^{\frac{s}{2} } }      \right)^{\frac{1-s}{2+s}} \\
							&\leq	c \, | \lambda |^{ - \frac{3s(1+s)}{2(2+s)}}  \left(   \avEW{  |  \Delta  |^{s/2} }   \right)^{2\beta}
	\end{align}
	where Aizenman and Warzel used their Theorem 2.1 in the second inequality together with the definition $\beta := \frac{1-s}{2(2+s)}$. The application of the resolvent identity yields
	\beq
		\Delta = \sum_{\mathbf{u}, \mathbf{w} \in C^{(n)}(\Omega)}   G_{\Omega}^{(J,K)}(\mathbf{x},\mathbf{u}, z) \sp{\dirac{u}}{(H_{\Omega} - H_{\Omega}^{(J,K)}) \dirac{w}} G_{\Omega}(\mathbf{w},\mathbf{y}, z).
	\eeq
	We thus arrive at
	\begin{align}
		\avEW{  |  \Delta  |^s } 	&\leq  	c \, | \lambda |^{ - \frac{3s(1+s)}{2(2+s)}}   \avEW{  \left|    \sum_{\mathbf{u}, \mathbf{w} \in C^{(n)}(\Omega)}   G_{\Omega}^{(J,K)}(\mathbf{x},\mathbf{u}, z) \sp{\dirac{u}}{(H_{\Omega} - H_{\Omega}^{(J,K)}) \dirac{w}} G_{\Omega}(\mathbf{w},\mathbf{y}, z)    \right|^{s/2} }^{2\beta} \\
							&\leq	c \, | \lambda |^{ - \frac{3s(1+s)}{2(2+s)}} \sum_{\mathbf{w}}   \avEW{  \left|    \sum_{\mathbf{u}}   G_{\Omega}^{(J,K)}(\mathbf{x},\mathbf{u}, z) \sp{\dirac{u}}{(H_{\Omega} - H_{\Omega}^{(J,K)}) \dirac{w}}   \right|^{s/2}     \left|  G_{\Omega}(\mathbf{w},\mathbf{y}, z)    \right|^{s/2} }^{2\beta} \\
							&\leq	c \, | \lambda |^{ - \frac{3s(1+s)}{2(2+s)}} \sum_{\mathbf{w}}   \avEW{ \left|  \sum_{\mathbf{u}}   G_{\Omega}^{(J,K)}(\mathbf{x},\mathbf{u}, z) \sp{\dirac{u}}{(H_{\Omega} - H_{\Omega}^{(J,K)}) \dirac{w}}   \right|^{s}  }^{\beta}     \avEW{ \left|  G_{\Omega}(\mathbf{w},\mathbf{y}, z)    \right|^{s}  }^{\beta}.
	\end{align}
	The last inequality follows from the Cauchy-Schwarz inequality. Theorem 2.1 from the original paper~\cite{AizenmanWarzel2009} gives
	\beq
		\avEW{ \left|  G_{\Omega}(\mathbf{w},\mathbf{y}, z)    \right|^{s}  }^{\beta} \leq \left( C_{s} \frac{K E_{0}}{ | \lambda |^s E_{0}^s }    \right)^{\beta}
	\eeq
	(with constants $C_{s}, K, E_{0}$) so that  
	\beq
		\avEW{  |  \Delta  |^s } \leq  \frac{\tilde{C}}{\lambda^s}   \sum_{\mathbf{w}}   \avEW{  \sum_{\mathbf{u}}  \left| G_{\Omega}^{(J,K)}(\mathbf{x},\mathbf{u}, z) \right|^{s} \, \left|  \sp{\dirac{u}}{(H_{\Omega} - H_{\Omega}^{(J,K)}) \dirac{w}}   \right|^{s}  }^{\beta}
	\eeq
	(the constant $\tilde{C}$ collects constant factors). Next we note that $\sp{\dirac{u}}{(H_{\Omega} - H_{\Omega}^{(J,K)}) \dirac{w}}$ is not a random variable. Therefore,
	\begin{align}
		\avEW{  |  \Delta  |^s } 	& \leq  \frac{\tilde{C}}{| \lambda |^s}  \sum_{\mathbf{w}}    \left|     \sum_{\mathbf{u}}  \left| \sp{\dirac{u}}{(H_{\Omega} - H_{\Omega}^{(J,K)}) \dirac{w}} \right|^s   \avEW{  \left| G_{\Omega}^{(J,K)}(\mathbf{x},\mathbf{u}, z)   \right|^{s}  }  \right|^{\beta}  \\
							& \leq	\frac{\tilde{C}}{| \lambda |^s}  \sum_{\mathbf{w}}    \left|     \sum_{\mathbf{u}}  \left| \sp{\dirac{u}}{(H_{\Omega} - H_{\Omega}^{(J,K)}) \dirac{w}} \right|^s   \tilde{A} e^{-\distHJK{x}{u} / \xi}  \right|^{\beta} \\
							& \leq	\frac{\tilde{C} \tilde{A}}{| \lambda |^s}  \sum_{\mathbf{w}}     \sum_{\mathbf{u}}  \left| \sp{\dirac{u}}{(H_{\Omega} - H_{\Omega}^{(J,K)}) \dirac{w}} \right|^{s \beta}  e^{- \frac{\beta}{\xi} \distHJK{x}{u} }  \\
							& \leq	\frac{\tilde{C} \tilde{A}}{| \lambda |^s}  \sum_{\mathbf{u} \in \supp{H_{\Omega} - H_{\Omega}^{(J,K)} }}     \sum_{\mathbf{w}} \left| \sp{\dirac{u}}{(H_{\Omega} - H_{\Omega}^{(J,K)}) \dirac{w}} \right|^{s \beta}   e^{- \frac{\beta}{\xi} \distHJK{x}{u} }
	\end{align}
	where we have used Theorem 5.2 in~Ref.~\onlinecite{AizenmanWarzel2009} and $| a + b |^{\beta} \leq | a |^{\beta} + | b |^{\beta}$, $\beta \in (0,1)$. The hopping we consider is finite-ranged. Set $m \in \mathbb{N}$ equal to the maximal hopping range and let $B_{m}(\mathbf{u})$ denote the ball around $\mathbf{u}$ with 1-norm radius $m$. Hence, for a fixed $\mathbf{u} \in C^{(n)}(\Omega)$ there are at most $| B_{m}(\mathbf{u}) | $ nonzero matrix elements of the operator $H_{\Omega}$ and thus of $H_{\Omega} - H_{\Omega}^{(J,K)}$. Therefore, we observe that the bound $| B_{m}(\mathbf{u}) | < (2m)^{nd}$ leads to
	\beq
		\avEW{  |  \Delta  |^s }	\leq	\frac{\tilde{C} \tilde{A}}{| \lambda |^s} (h_{\max})^{s\beta} (2m)^{nd}  \sum_{\mathbf{u} \in \supp{H_{\Omega} - H_{\Omega}^{(J,K)} }}      e^{- \frac{\beta}{\xi} \distHJK{x}{u} }
	\eeq
	with $h_{\max} := \max_{\mathbf{n}, \mathbf{m}} | H(\mathbf{n},\mathbf{m}) |$. From Eq.~(5.19) in the original paper~\cite{AizenmanWarzel2009} we know that 
	\beq
		\min_{\mathbf{u} \in   \supp{H_{\Omega} - H_{\Omega}^{(J,K)} }}   \distHJK{x}{u}  -   \left(  l(\mathbf{x}) - m    \right) \geq 0 .
	\eeq
	It follows that
	\begin{align}
		\sum_{\mathbf{u} \in \supp{H_{\Omega} - H_{\Omega}^{(J,K)} }}      e^{- \frac{\beta}{\xi} \distHJK{x}{u} }	
		 & = e^{- \frac{\beta}{\xi} ( l(\mathbf{x}) - m ) }  \sum_{\mathbf{u} \in \supp{H_{\Omega} - H_{\Omega}^{(J,K)} }}      e^{- \frac{\beta}{\xi} \left( \distHJK{x}{u} - (l(\mathbf{x}) - m)  \right) } \nn \\
		 & \leq e^{- \frac{\beta}{\xi} ( l(\mathbf{x}) - m ) }  \sum_{\mathbf{u} \in \supp{H_{\Omega} - H_{\Omega}^{(J,K)} }}      e^{- \frac{\beta}{\xi} \left( \distHJK{x}{u} - \min_{\mathbf{u} \in   \supp{H_{\Omega} - H_{\Omega}^{(J,K)} }}   \distHJK{x}{u}  \right) }  \nn \\
		& \leq	e^{- \frac{\beta}{\xi} ( l(\mathbf{x}) - m ) }  \sum_{\mathbf{u} \in C^{(n)}(\Omega) }      e^{- \frac{\beta}{\xi} \distHJK{a}{u} }  ,
	\end{align}
	where $\mathbf{a} \in \mathbb{Z}^{nd}$ realizes the minimum $\min_{\mathbf{u} \in   \supp{H_{\Omega} - H_{\Omega}^{(J,K)} }}   \distHJK{x}{u}$. The use of the relation $\distHJK{\mathbf{x}}{\mathbf{y}} \geq \distH{\mathbf{x}}{\mathbf{y}}$ yields
	\begin{align}	
	\sum_{\mathbf{u} \in \supp{H_{\Omega} - H_{\Omega}^{(J,K)} }}      e^{- \frac{\beta}{\xi} \distHJK{x}{u} }
		& \leq	e^{- \frac{\beta}{\xi} ( l(\mathbf{x}) - m ) }  \sum_{\mathbf{u} \in C^{(n)}(\Omega) }      e^{- \frac{\beta}{\xi} \distH{a}{u} }  \nn \\
		& \leq	e^{- \frac{\beta}{\xi} ( l(\mathbf{x}) - m ) }  C(n,d) \left( \frac{\xi}{\beta} \right)^{nd}.
	\end{align}	
	The last inequality is a consequence of Lemma A.3 in~Ref.~\onlinecite{AizenmanWarzel2009}. After this little detour we conclude
	\beq\label{temp3.in.pf.of.thm5.3}
		\avEW{  |  \Delta  |^s }	\leq	\frac{B}{| \lambda |^s}   e^{- \frac{\beta}{\xi} l(\mathbf{x}) }
	\eeq
	if we set
	\beq
		B :=  \tilde{C} \tilde{A} (h_{\max})^{s\beta} (2m)^{nd}    C(n,d) \left( \frac{\xi}{\beta} \right)^{nd}  e^{ \frac{\beta}{\xi} m  } < \infty
	\eeq
	Putting together the Eq.s~\eqref{temp1.in.pf.of.thm5.3}, \eqref{temp2.in.pf.of.thm5.3} and \eqref{temp3.in.pf.of.thm5.3} we arrive at
	\beq
		\avEW{  |  G_{\Omega}(\mathbf{x},\mathbf{y})  |^s} \leq A e^{-\frac{1}{\xi} \min\{  \distH{\mathbf{x}}{\mathbf{y}}, l(\mathbf{x})  \}  }
	\eeq
	This proves the theorem because we assumed without loss of generality $l(\mathbf{x}) \geq l(\mathbf{y})$ at the beginning of the proof.
\end{proof}

Set
	\beq\label{def.B}
		B_{s}^{(n)}(L_{k}) := | \partial^{(m)} \Lambda_{L_{k}} | \sup_{\substack{I \subseteq \AR \\ | I | \geq 1}} \sup_{\tilde{\Omega} \subseteq \Lambda_{L_{k}}}  \sum_{y \in \partial^{(m)} \Lambda_{L_{k}}} \sum_{\substack{ \mathbf{x} \in C^{(n)}_{r_{L_{k} - 4m}}(\tilde{\Omega};0)  \\  \mathbf{y} \in C^{(n)}_{r_{L_{k}}}(\tilde{\Omega};y) }}    \avEW{  |  G^{(n)}_{\tilde{\Omega}}(\mathbf{x},\mathbf{y})  |^s}
	\eeq
with
	\beq\label{def.wide.inner.bdry}
		\partial^{(m)} M := \{   w \in M  |  \min_{q \in \partial^{(-)}M } \| q - w  \|_{1} \leq m-1    \}
	\eeq
where $\partial^{(-)}M$ denotes the inner boundary of $M \subset \mathbb{Z}^{d}$. Further we define the sequence $\{ L_{k} \}_{k\geq 0}$ recursively: $L_{k+1} := 2(L_{k} + 1)$. The starting point $L_{0}$ will be specified later.

\begin{theorem}[Theorem 6.1]\label{theorem.6.1}
	Assume that the $(n-1)$-defect system satisfies the fractional moment criterium~\eqref{np.aizenman.criterium} for all $\lambda \in (\lambda_{0}, \infty)$, $\lambda_{0} < \infty$. Then there exist some $s \in (0,1)$, $a,A,p < \infty$, and $\nu > 0$ such that
	\beq
		B_{s}^{(n)}(L_{k+1})  \leq  \frac{a}{| \lambda |^s} B_{s}^{(n)}(L_{k})^2  +  A L_{k+1}^{2p} e^{-2 \nu L_{k}},
	\eeq
	for all $\lambda \in (\lambda_{0}, \infty)$ and all $k \in \mathbb{N}_{0}$.
\end{theorem}

\begin{proof}
	Set
	\beq
		\tilde{B}_{s}^{(n)}(L_{k+1}) := | \partial^{(m)} \Lambda_{L_{k+1}} | \sup_{\substack{I \subseteq \AR \\ | I | \geq 1}} \sup_{\tilde{\Omega} \subseteq \Lambda_{L_{k+1}}}  \sum_{y \in \partial^{(m)} \Lambda_{L_{k+1}}} \sum_{\substack{ \mathbf{x} \in C^{(n)}_{r_{L_{k} - 4m}}(\tilde{\Omega};0)  \\  \mathbf{y} \in C^{(n)}_{r_{L_{k}}}(\tilde{\Omega};y) }}    \avEW{  |  G^{(n)}_{\tilde{\Omega}}(\mathbf{x},\mathbf{y})  |^s}.
	\eeq
	We follow the strategy of Aizenman and Warzel and write $B_{s}^{(n)}(L_{k+1})$ in the form
	\beq\label{splitting.B(n)s.up}
		B_{s}^{(n)}(L_{k+1}) = \tilde{B}_{s}^{(n)}(L_{k+1}) + B_{s}^{(n)}(L_{k+1}) - \tilde{B}_{s}^{(n)}(L_{k+1}) 
	\eeq
	and estimate $\tilde{B}_{s}^{(n)}(L_{k+1})$ and the difference appearing on the RHS separately. We start with the estimation of $B_{s}^{(n)}(L_{k+1}) - \tilde{B}_{s}^{(n)}(L_{k+1})$:
	\beq\label{B(n)s.Difference}
		B_{s}^{(n)}(L_{k+1}) - \tilde{B}_{s}^{(n)}(L_{k+1}) 
		=	| \partial^{(m)} \Lambda_{L_{k+1}} | \sup_{\substack{I \subseteq \AR \\ | I | \geq 1}} \sup_{\tilde{\Omega} \subseteq \Lambda_{L_{k+1}}}   \sum_{y \in \partial^{(m)} \Lambda_{L_{k+1}}}  \Delta
	\eeq
	with
	\begin{align}
		\Delta 
		&:=	\sum_{\substack{ \mathbf{x} \in C^{(n)}_{r_{L_{k+1} - 4m}}(\tilde{\Omega};0)  \\  \mathbf{y} \in C^{(n)}_{r_{L_{k+1}}}(\tilde{\Omega};y) }}    \avEW{  |  G^{(n)}_{\tilde{\Omega}}(\mathbf{x},\mathbf{y})  |^s}     -    \sum_{\substack{ \mathbf{x} \in C^{(n)}_{r_{L_{k} - 4m}}(\tilde{\Omega};0)  \\  \mathbf{y} \in C^{(n)}_{r_{L_{k}}}(\tilde{\Omega};y) }}    \avEW{  |  G^{(n)}_{\tilde{\Omega}}(\mathbf{x},\mathbf{y})  |^s} \nn \\
		&=	\sum_{\mathbf{x} \in C^{(n)}_{r_{L_{k} - 4m}}(\tilde{\Omega};0) }  \left(    \sum_{\mathbf{y} \in C^{(n)}_{r_{L_{k+1}}}(\tilde{\Omega};y) }  \avEW{ ...}     -      \sum_{\mathbf{y} \in C^{(n)}_{r_{L_{k}}}(\tilde{\Omega};y) }  \avEW{ ...}    \right)        +        \sum_{\mathbf{x} \in C^{(n)}_{r_{L_{k+1}-4m}}(\tilde{\Omega};0) \backslash C^{(n)}_{r_{L_{k}-4m}}(\tilde{\Omega};0)}   \sum_{\mathbf{y} \in C^{(n)}_{r_{L_{k+1}}}(\tilde{\Omega};y)}  \avEW{...} .
	\end{align}
	The quantity $\Delta$ can be bounded as follows:
	\begin{align}
		\Delta 
		&=	\sum_{\substack{ \mathbf{x} \in C^{(n)}_{r_{L_{k} - 4m}}(\tilde{\Omega};0)  \\  \mathbf{y} \in C^{(n)}_{r_{L_{k+1}}}(\tilde{\Omega};y)  \backslash  C^{(n)}_{r_{L_{k}}}(\tilde{\Omega};y)     }}      \avEW{...}
				+  \sum_{\substack{  \mathbf{x} \in C^{(n)}_{r_{L_{k+1}-4m}}(\tilde{\Omega};0)  \backslash  C^{(n)}_{r_{L_{k}-4m}}(\tilde{\Omega};0)    \\    \mathbf{y} \in C^{(n)}_{r_{L_{k+1}}}(\tilde{\Omega};y)    }}      \avEW{...}  \nn \\
		 &\leq	\Xi_{1}  + \Xi_{2}
	\end{align}
	with
	\begin{align}
		\Xi_{1} 
		&:=  \sup_{| x-y | \geq 2L_{k+1} - m}  \sum_{\substack{ \mathbf{y} \in C^{(n)}_{r_{L_{k} - 4m}}(\tilde{\Omega};y)  \\  \mathbf{x} \in C^{(n)}(\tilde{\Omega};x)  \backslash  C^{(n)}_{r_{L_{k}}}(\tilde{\Omega};x)     }}      \avEW{...}  \\
		\Xi_{2}
		&:=   \sup_{| x-y | \geq 2L_{k+1} - m}  \sum_{\substack{ \mathbf{y} \in C^{(n)}_{r_{L_{k+1}}}(\tilde{\Omega};y)  \\  \mathbf{x} \in C^{(n)}(\tilde{\Omega};x)  \backslash  C^{(n)}_{r_{L_{k}-4m}}(\tilde{\Omega};x)     }}      \avEW{...} .
	\end{align}
	The last inequality a consequence of $y \in \partial^{(m)} \Lambda_{L_{k+1}}$. We proceed by the treatment of $\Xi_{1}$. According to~(A.8) in~Ref.~\onlinecite{AizenmanWarzel2009}
	\beq
		\distH{x}{y} \geq | x-y | - (L_{k} - 4m) \geq \frac{L_{k}}{n-1}
	\eeq
	because
	\beq
		| x-y | \geq 2L_{k+1} - m = 2^2 (L_{k} + 1) - m.
	\eeq
	At the same time, $L_{k} / (n-1)$ serves as a lower bound for $l(\mathbf{x})$, so that $\distH{x}{y}$ and $\max\{ l(\mathbf{x}), l(\mathbf{y}) \}$ are both lower bounded by $L_{k} / (n-1)$. Theorem~\ref{theorem.5.3} thus yields
	\beq
		\avEW{  |  G^{(n)}_{\tilde{\Omega}}(\mathbf{x},\mathbf{y})  |^s} \leq A e^{-\frac{L_{k}}{\xi(n-1)}}
	\eeq
	for the summands appearing in the definition of $\Xi_{1}$. The number of summands can be bounded with the help of the observation that
	\beq
		| C^{(n)}_{r_{L_{k}-4m}}(\tilde{\Omega},y)  | \leq n (4 r_{L_{k}-4m})^{d(n-1)}, \ | C^{(n)}(\tilde{\Omega},x) | \leq n | \tilde{\Omega} |^{n-1}.
	\eeq
	The combination of all these estimation gives
	\beq
		\Xi_{1} \leq \tilde{A}_{1} (r_{L_{k}-4m})^{d(n-1)}  | \tilde{\Omega} |^{n-1}   e^{-\frac{L_{k}}{\xi(n-1)}}.
	\eeq
	Similarly, we get
	\beq
		\Xi_{2} \leq \tilde{A}_{2} (r_{L_{k+1}})^{d(n-1)}  | \tilde{\Omega} |^{n-1}   e^{-\frac{L_{k}}{\xi(n-1)}}.
	\eeq
	Note that for example
	\beq
		(r_{L_{k+1}})^{d(n-1)} \left( \frac{L_{k+1}}{2} \right)^{d(n-1)} <  \left( (2 L_{k+1})^d \right)^{(n-1)} = | \Lambda_{L_{k+1}} |^{n-1}.
	\eeq
	We conclude that
	\beq
		\Delta \leq \tilde{\tilde{A}} \left( L_{k+1}^{dn-1} \right)^{2}  e^{-\frac{L_{k}}{\xi(n-1)}}
	\eeq
	after convenient definitions for $\tilde{\tilde{A}}$ and $\xi$. Going back to~\eqref{B(n)s.Difference}, we find
	\beq\label{pf.of.thm.6.1.temp.1}
		B_{s}^{(n)}(L_{k+1}) - \tilde{B}_{s}^{(n)}(L_{k+1}) 
		\leq	| \partial^{(m)} \Lambda_{L_{k+1}} |^2 \, A \left( L_{k+1}^{dn-1} \right)^{2}  e^{-\frac{L_{k}}{\xi(n-1)}}
	\eeq
	after the convenient redefinition of $A$. One of the factors $| \partial^{(m)} \Lambda_{L_{k+1}} |$ emerges from the replacement of the sum $\sum_{y \in \partial^{(m)} \Lambda_{L_{k+1}}}$ in~\eqref{B(n)s.Difference}. To conclude the proof of the theorem we need to estimate $\tilde{B}_{s}^{(n)}(L_{k+1})$ (cf.~\eqref{splitting.B(n)s.up}). The application of Lemma~\ref{Lemma.6.3} yields the upper bound
	\begin{align}
		\tilde{B}_{s}^{(n)}(L_{k+1})						\leq &	\frac{C}{| \lambda |^s} \, h_{\max}^{2s} \, | \max \mathrm{supp} H_0 |^2             | \partial^{(m)} \Lambda_{L_{k+1}} | \sup_{\substack{I \subseteq \AR \\ | I | \geq 1}} \sup_{\tilde{\Omega} \subseteq \Lambda_{L_{k+1}}}       | \partial^{(m)} \Lambda_{L_{k+1}} |  \\ & \times  \sup_{y \in \partial^{(m)} \Lambda_{L_{k+1}}}     
		\left(    \sum_{\mathbf{x} \in C^{(n)}_{r_{L_{k}}}(\tilde{\Omega};0)} \sum_{u \in \partial^{(m)}(\Lambda_{L_{k}} \cap \tilde{\Omega})} \sum_{\mathbf{w} \in C^{(n)}(\Lambda_{L_{k}} \cap \tilde{\Omega};u)}   \avEW{  |  G^{(n)}_{V}(\mathbf{x},\mathbf{w})  |^s}  \right)    \\       
		&  \times \left(    \sum_{\mathbf{y} \in C^{(n)}_{r_{L_{k}}}(\tilde{\Omega};y)} \sum_{\bar{u} \in \partial^{(m)}(\Lambda_{L_{k}} \cap \tilde{\Omega})} \sum_{\mathbf{v}' \in C^{(n)}(\Lambda_{L_{k}}(y) \cap \tilde{\Omega};\bar{u})}   \avEW{  |  G^{(n)}_{W}(\mathbf{v}',\mathbf{y})  |^s}  \right)  
	\end{align}
	with $V = \Lambda_{L_{k}} \cap \tilde{\Omega}$ and $W = \Lambda_{L_{k}}(y) \cap \tilde{\Omega}$. We continue with
	\begin{align}
		\tilde{B}_{s}^{(n)}(L_{k+1})
		\leq &	\frac{C}{| \lambda |^s} \, h_{\max}^{2s} \, | \max \mathrm{supp} H_0 |^2  | \partial^{(m)} \Lambda_{L_{k+1}} |^2 \sup_{\substack{I \subseteq \AR \\ | I | \geq 1}} \Bigg\{ \\
		& \times \sup_{\tilde{\Omega} \subseteq \Lambda_{L_{k}}}         \bigg\{   \sum_{u \in \partial^{(m)}(\Lambda_{L_{k}} \cap \tilde{\Omega})}   \sum_{\substack{\mathbf{x} \in C^{(n)}_{r_{L_{k}}}(\tilde{\Omega};0)   \\   \mathbf{w} \in C^{(n)}(\Lambda_{L_{k}} \cap \tilde{\Omega};u)   }}    \avEW{  |  G^{(n)}_{\tilde{\Omega}}(\mathbf{x},\mathbf{w})  |^s}     \bigg\} \\
		& \times	\sup_{y \in \partial^{(m)} \Lambda_{L_{k+1}}} \sup_{\tilde{\Omega} \subseteq \Lambda_{L_{k}}}         \bigg\{   \sum_{\bar{u} \in \partial^{(m)}(\Lambda_{L_{k}} \cap \tilde{\Omega})}   \sum_{\substack{\mathbf{y} \in C^{(n)}_{r_{L_{k}}}(\tilde{\Omega};y)   \\   \mathbf{v}' \in C^{(n)}(\Lambda_{L_{k}}(y) \cap \tilde{\Omega};\bar{u})   }}    \avEW{  |  G^{(n)}_{\tilde{\Omega}}(\mathbf{x},\mathbf{w})  |^s}     \bigg\} \Bigg\}.
	\end{align}
	The translation invariance (simultaneous translation of all particles) allows to shift $y \in \mathbb{Z}^d$ to the origin. Therefore,
	\begin{align}
		\tilde{B}_{s}^{(n)}(L_{k+1})
		\leq &	\frac{C}{| \lambda |^s} \, h_{\max}^{2s} \, | \max \mathrm{supp} H_0 |^2  \eta^2 | \partial^{(m)} \Lambda_{L_{k}} |^2 \\
		& \times \Bigg(      \sup_{\substack{I \subseteq \AR \\ | I | \geq 1}} \sup_{\tilde{\Omega} \subseteq \Lambda_{L_{k}}}       \sum_{u \in \partial^{(m)}(\Lambda_{L_{k}} \cap \tilde{\Omega})}   \sum_{\substack{\mathbf{x} \in C^{(n)}_{r_{L_{k}}}(\tilde{\Omega};0)   \\   \mathbf{w} \in C^{(n)}(\Lambda_{L_{k}} \cap \tilde{\Omega};u)   }}    \avEW{  |  G^{(n)}_{\Omega}(\mathbf{x},\mathbf{w})  |^s}     \Bigg)^2
	\end{align}
	where $\eta < \infty$ is defined such that
	\beq
		| \partial^{(m)} \Lambda_{L_{k+1}} |  \leq  \eta | \partial^{(m)} \Lambda_{L_{k}} |
	\eeq
	for all $k$. We can rewrite this bound in terms of the function $k_{s}(\Lambda_{L_{k}}, r_{k}, r_{k})$ (its definition can be found in Corollary 5.4 of the original paper):
	\begin{align}
		\tilde{B}_{s}^{(n)}(L_{k+1})
		\leq &	\eta^2 \frac{C}{| \lambda |^s} \, h_{\max}^{2s} \, | \max \mathrm{supp} H_0 |^2   \left(   B_{s}^{(n)}(L_{k})   + k_{s}(\Lambda_{L_{k}}, r_{k}, r_{k})  | \partial^{(m)}( \Lambda_{L_{k}} \cap \Omega ) |^2  \right)^2 \\
		 \leq &	  \eta^2 \frac{C}{| \lambda |^s} \, h_{\max}^{2s} \, | \max \mathrm{supp} H_0 |^2   \left(   B_{s}^{(n)}(L_{k})   + k_{s}(\Lambda_{L_{k}}, r_{k}, r_{k})  ( m | \partial \Lambda_{L_{k}} | )^2  \right)^2
	\end{align}
	where $m$ denotes the hopping range. The application of the Corollary 5.4 thus gives
	\beq\label{pf.of.thm.6.1.temp.2}
		\tilde{B}_{s}^{(n)}(L_{k+1}) \leq   \eta^2 \frac{C}{| \lambda |^s} \, h_{\max}^{2s} \, | \max \mathrm{supp} H_0 |^2   \left(   B_{s}^{(n)}(L_{k}) + A^2 L_{k}^{2dn-2} e^{-\frac{L_{k}}{(n-1) \xi}}   \right)^2.
	\eeq
	After the appropriate redefinitions of the constants the combination of the bounds~\eqref{pf.of.thm.6.1.temp.1} and~\eqref{pf.of.thm.6.1.temp.2} yields the upper bound of the theorem.
\end{proof}

\begin{lemma}[Lemma 6.3]\label{Lemma.6.3}
	Let $\Omega \subset \mathbb{Z}^d$ and $V,W \in \Omega$ with $\text{dist}(V,W) \geq 2$. Then for all $\mathbf{x} \in \mathcal{C}^{(n)}(V)$ and $\mathbf{y} \in \mathcal{C}^{(n)}(W)$
	\beqa
		\EW{  |  G_{\Omega}^{(n)}(\mathbf{x},\mathbf{y};z)  |^s  }   &\leq&   \frac{C}{| \lambda |^s} \, h_{\max}^{2s} \, | \max \mathrm{supp} H_0 |^2  \sum_{u \in \partial^{(m)} V}  \sum_{\mathbf{w} \in \mathcal{C}^{(n)}(V;u)  }   \EW{  |  G_{V}^{(n)}(\mathbf{x},\mathbf{w};z)  |^s  } \nn \\
												           & &       \times   \sum_{v \in \partial^{(m)} W}  \sum_{\mathbf{v} \in \mathcal{C}^{(n)}(W;v)  }   \EW{  |  G_{W}^{(n)}(\mathbf{v},\mathbf{y};z)  |^s  } 
	\eeqa
	where the constant $C = C(s,d) < \infty$ is independent of $(\lambda, \mathbf{\alpha}) \in \mathbb{R}^p$ (cf.~\eqref{def.wide.inner.bdry} for the definition of $\partial^{(m)}(\cdots)$).
\end{lemma}

\begin{proof}
	To remove all the terms in $H_{\Omega}$ that connect sites in $C^{(n)}(V)$ with sites in the complement of $C^{(n)}(V)$ (i.e., $C^{(n)}(\Omega; \Omega \backslash V)$)  and the analog for $V \leftrightarrow W$ we define the boundary-strip operators
	\begin{align}
		\Gamma_{\Omega \backslash V}^{V} &:=	H_{\Omega} - H_{V} \oplus H_{\Omega \backslash V} \\ 
		\Gamma_{\Omega \backslash W}^{W} &:=	H_{\Omega} - H_{W} \oplus H_{\Omega \backslash W}.
	\end{align}
	By assumption $\mathbf{x} \in C^{(n)}(V), \mathbf{y} \in C^{(n)}(W)$ and $V \cap W = \emptyset$. Note that for $\mathbf{x} \in \mathcal{C}^{(n)}(V)$ and $\mathbf{y} \in \mathcal{C}^{(n)}(W) \subseteq C^{(n)}(\Omega; \Omega \backslash V)$ we have for example
	\beq
		\sp{\dirac{x}}{ (H_{V} \oplus H_{\Omega \backslash V} - z)^{-1} \dirac{y} } = 0.
	\eeq
	Therefore, the resolvent identity yields
	\begin{align}
		G_{\Omega}(\mathbf{x}, \mathbf{y};z)	=&	\sp{\dirac{x}}{(H_{\Omega}-z)^{-1} \dirac{y}} \nn \\
										=&	\sp{\dirac{x}}{ -(H_{V} \oplus H_{\Omega \backslash V} - z)^{-1} \Gamma_{\Omega \backslash V}^{V}  (H_{\Omega}-z)^{-1} \dirac{y}} \nn \\
										=&	-\sp{\dirac{x}}{(H_{V} \oplus H_{\Omega \backslash V} - z)^{-1} \Gamma_{\Omega \backslash V}^{V}  (H_{\Omega}-z)^{-1} \Gamma_{\Omega \backslash W}^{W}  (H_{W} \oplus H_{\Omega \backslash W} - z)^{-1}   \dirac{y}} \nn \\
										=&	-\sum_{\substack{   \mathbf{w} \in C^{(n)}(V) \\ \mathbf{w}' \in C^{(n)}(\Omega, \Omega \backslash V)   }}  \sum_{\substack{   \mathbf{v} \in C^{(n)}(W) \\ \mathbf{v}' \in C^{(n)}(\Omega, \Omega \backslash W)   }}
											\sp{\dirac{x}}{(H_{V} \oplus H_{\Omega \backslash V} - z)^{-1} \dirac{w}}   \sp{\dirac{w}}{\Gamma_{\Omega \backslash V}^{V} \dirac{w'}}   \\ & \times  \sp{\dirac{w'}}{(H_{\Omega}-z)^{-1} \dirac{v}}  \sp{\dirac{v}}{\Gamma_{\Omega \backslash W}^{W} \dirac{v'}}  \sp{\dirac{v'}}{(H_{W} \oplus H_{\Omega \backslash W} - z)^{-1}  \dirac{y}} \nn \\
										=&	-\sum_{\substack{   \mathbf{w} \in C^{(n)}(V) \\ \mathbf{w}' \in C^{(n)}(\Omega, \Omega \backslash V)   }}  \sum_{\substack{   \mathbf{v} \in C^{(n)}(W) \\ \mathbf{v}' \in C^{(n)}(\Omega, \Omega \backslash W)   }}
											G_{V}^{(n)}(\mathbf{x},\mathbf{w};z)  \Gamma_{\Omega \backslash V}^{V}(\mathbf{w},\mathbf{w'})  G_{\Omega}^{(n)}(\mathbf{w'},\mathbf{v};z)  \Gamma_{\Omega \backslash W}^{W}(\mathbf{v},\mathbf{v'})  G_{W}^{(n)}(\mathbf{v'},\mathbf{y};z)
 	\end{align}	
	Using Theorem 2.1 of the original paper~\cite{AizenmanWarzel2009} we thus get
	\begin{align}
		\EW{|G_{\Omega}(\mathbf{x}, \mathbf{y};z) |^s }	\leq&		\sum_{\substack{   \mathbf{w} \in C^{(n)}(V) \\ \mathbf{w}' \in C^{(n)}(\Omega, \Omega \backslash V)   }}  \sum_{\substack{   \mathbf{v} \in C^{(n)}(W) \\ \mathbf{v}' \in C^{(n)}(\Omega, \Omega \backslash W)   }}
														| \Gamma_{\Omega \backslash V}^{V}(\mathbf{w},\mathbf{w'}) |^s | \Gamma_{\Omega \backslash W}^{W}(\mathbf{v},\mathbf{v'}) |^s \nn \\ &  \times \EW{ \left| G_{V}^{(n)}(\mathbf{x},\mathbf{w};z)    G_{\Omega}^{(n)}(\mathbf{w'},\mathbf{v};z)    G_{W}^{(n)}(\mathbf{v'},\mathbf{y};z) \right|^s  } \nn \\									
 												\leq&		\frac{C}{| \lambda |^s} \sum_{\substack{   \mathbf{w} \in C^{(n)}(V) \\ \mathbf{w}' \in C^{(n)}(\Omega, \Omega \backslash V)   }}  \sum_{\substack{   \mathbf{v} \in C^{(n)}(W) \\ \mathbf{v}' \in C^{(n)}(\Omega, \Omega \backslash W)   }}
														| \Gamma_{\Omega \backslash V}^{V}(\mathbf{w},\mathbf{w'}) |^s | \Gamma_{\Omega \backslash W}^{W}(\mathbf{v},\mathbf{v'}) |^s \nn \\ &  \times \EW{ \left|  G_{V}^{(n)}(\mathbf{x},\mathbf{w};z)  \right|^s }  \EW{ \left|  G_{W}^{(n)}(\mathbf{v'},\mathbf{y};z)  \right|^s }.
	\end{align}	
	We observe that
	\beq \label{pf.Lemma.6.3.global.1}
		\{  (\mathbf{q}, \mathbf{q'}) \in \mathbb{Z}^{nd} \times \mathbb{Z}^{nd}  | \sp{\dirac{q}}{\Gamma_{\Omega \backslash M}^M \dirac{q'}  }  \neq 0 \} \subseteq \overline{\partial^{(m)} M } \times \overline{\partial^{(m)} M }
	\eeq
	(see~\eqref{pf.Lemma.6.3.global.1asdsd}). Therefore, $\Gamma_{\Omega \backslash W}^{W}(\mathbf{w},\mathbf{w'}) \neq 0$ for $\mathbf{w} \in \overline{\partial^{(m)} V } \cap C^{(n)}(V)$ implies $\mathbf{w'} \in \overline{\partial^{(m)} V } \backslash C^{(n)}(V)$, and $\Gamma_{\Omega \backslash W}^{W}(\mathbf{v},\mathbf{v'}) \neq 0$ for $\mathbf{v'} \in \overline{\partial^{(m)} W } \cap C^{(n)}(W)$ implies $\mathbf{v} \in \overline{\partial^{(m)} W } \backslash C^{(n)}(W)$. We conclude that
	\begin{align}
		\EW{|G_{\Omega}(\mathbf{x}, \mathbf{y};z) |^s }	\leq&		\frac{C}{| \lambda |^s} \sum_{\substack{   \mathbf{w} \in  \overline{\partial^{(m)} V} \cap C^{(n)}(V) \\ \mathbf{w}' \in \overline{\partial^{(m)} V} \backslash C^{(n)}(V)   }}  \sum_{\substack{   \mathbf{v} \in \overline{\partial^{(m)} W} \backslash C^{(n)}(W) \\ \mathbf{v}' \in \overline{\partial^{(m)} W} \cap C^{(n)}(W)   }}
														| \Gamma_{\Omega \backslash V}^{V}(\mathbf{w},\mathbf{w'}) |^s | \Gamma_{\Omega \backslash W}^{W}(\mathbf{v},\mathbf{v'}) |^s \nn \\ &  \times \EW{ \left| G_{V}^{(n)}(\mathbf{x},\mathbf{w};z)  \right|^s}  \EW{ \left| G_{W}^{(n)}(\mathbf{v'},\mathbf{y};z)  \right|^s }.
	\end{align}	
	Note that the product of the two expectation values is independent of $\mathbf{w'}$ and $\mathbf{v}$. Thus, for each $\mathbf{w}$ and $\mathbf{v'}$ we have at most $| \max \mathrm{supp} H_0 |$ non-vanishing terms that are upper bounded by $h_{\max}$ so that
	\begin{align}
		\EW{|G_{\Omega}(\mathbf{x}, \mathbf{y};z) |^s }	\leq&		\frac{C}{| \lambda |^s} h_{\max}^{2s} | \max \mathrm{supp} H_0 |^2   \sum_{   \mathbf{w} \in  \overline{\partial^{(m)} V} \cap C^{(n)}(V) }  \sum_{  \mathbf{v}' \in \overline{\partial^{(m)} W} \cap C^{(n)}(W)   } \EW{ \left| G_{V}^{(n)}(\mathbf{x},\mathbf{w};z)  \right|^s}  \EW{ \left| G_{W}^{(n)}(\mathbf{v'},\mathbf{y};z)  \right|^s }.
	\end{align}	
	Next, we rewrite the sums in the form
	\begin{align}\label{pf.Lemma.6.3.global.sum.rewriting}
		\sum_{   \mathbf{w} \in  \overline{\partial^{(m)} V} \cap C^{(n)}(V) }  |.....| \			&\leq 	\sum_{u \in \partial^{(m)} V} \sum_{\mathbf{w} \in C^{(n)}(V;u)} |.....|		 \nn \\
		\sum_{  \mathbf{v'} \in \overline{\partial^{(m)} W} \cap C^{(n)}(W)   }  |.....| \        	&\leq		\sum_{\bar{u} \in \partial^{(m)} W} \sum_{\mathbf{v'} \in C^{(n)}(W;\bar{u})} |.....|	
	\end{align}	
	to conclude the proof of the Lemma.
\end{proof}

\section{Notations}\label{Notations}

\begin{align*}
	\{ \mathbf{x} \} 
	&=	\left. \left\{ x_{j} \in \mathbb{Z}^d  \right|  j \in \{ 1, ..., n \} \right\} \\
	r_{L}
	&= \frac{L}{2} \\
	\Lambda_{L}
	&=	[-L,L] \cap \mathbb{Z}^d  \\
	C^{(n)}(\Omega)
	&= \{ \mathbf{x} = (x_{1}, ..., x_{n})  \, | \, x_{j} \in \Omega, \forall j \} \\
	C^{(n)}(\Omega, u)
	&= \{ \mathbf{x} \in C^{(n)}(\Omega)  \, | \, x_{j} = u \text{ for some j} \} \\
	C^{(n)}(\Omega, S)
	&= \{ \mathbf{x} \in C^{(n)}(\Omega)  \, | \, x_{j} \in S \text{ for some j} \} \\
	\partial^{(-)} M 
	&= \{   w \in M  |  \min_{q \in M^c } | q - w  | = 1    \}  \\
	\partial^{(m)} M 
	&= \{   w \in M  |  \min_{q \in \partial^{(-)}M } | q - w  | \leq m-1    \}  \\
	\overline{\partial^{(m)} M } 
	&= \left. \left\{   \mathbf{r} = (r_{1},...,r_{n}) \in C^{(n)}(M)  \right| \exists \mathbf{s} \in  \mathbb{Z}^{(nd)} \backslash C^{(n)}(M) \text{ such that }  \max\{ \| s_{i} - r_{j}   \|_{1} \leq m \, | \, i,j \in \{ 1,...,n \} \}  \right\}     \\
								& \ \ \ \ \cup \left.\left\{   \mathbf{r} = (r_{1},...,r_{n}) \in \mathbb{Z}^{(nd)} \backslash C^{(n)}(M) \right|  \exists \mathbf{s} \in  C^{(n)}(M) \text{ such that }  \max\{ \| s_{i} - r_{j}   \|_{1} \leq m \, | \, i,j \in \{ 1,...,n \} \}   \right\}   \\
	\mathrm{diam}(\mathbf{x})
	&=	\max_{j,k \in \{ 1, ..., n \}} | x_{j} - x_{k} | \\
	C^{(n)}_{r}(\Omega)
	&= \{ \mathbf{x} \in C^{(n)}(\Omega) \, | \, \mathrm{diam}(\mathbf{x}) \leq r \} \\
	l(\mathbf{x})
	&=	\max_{\substack{  J,K \\ J \dot{\cup} K = \{ 1, ..., n \}  }}  \min_{j \in J, \, k \in K} | x_{j} - x_{k} | \\
	\mathrm{dist}_{\mathcal{H}}(\mathbf{x}, \mathbf{y}) 
	&=	\max \left\{    \max_{1\leq i \leq k} \text{dist}(x_{i}, \{ \mathbf{y} \} ) ,  \max_{1\leq i \leq k} \text{dist}(\{ \mathbf{x} \}, y_{i} )    \right\} \\
	\distHJK{\mathbf{x}}{\mathbf{y}} 
	&=	\max \left\{    \mathrm{dist}_{\mathcal{H}}(\mathbf{x}_{J}, \mathbf{y}_{J}) ,  \mathrm{dist}_{\mathcal{H}}(\mathbf{x}_{K}, \mathbf{y}_{K})   \right\}  \\
	\max \mathrm{supp} H_{0} 
	&=	\bigcup_{\mathbf{k}} \left[ \mathrm{supp} H_0(\mathbf{k},\cdot) - \mathbf{k} \right] \\
	G_{\Omega}(\mathbf{x}, \mathbf{y};z)
	&=	\sp{\dirac{x}}{ (H_{\Omega} - z)^{-1} \dirac{y} } \\
	\avEW{|G_{\Omega}(\mathbf{x},\mathbf{y})|^s}
	&=	\frac{1}{| I |} \int_{I}  \EW{  |  G_{\Omega}^{(n)}(\mathbf{x},\mathbf{y};E)  |^s } dE
\end{align*}

\section{Explanation of the simulation procedure}\label{sec:numerics}

In this section, we describe the numerics behind the transition from the superfluid phase to the Bose glass phase for non-vanishing defect density. The procedure is a standard finite size scaling study in combination with quantum Monte Carlo simulations.

We consider the system with Hamiltonian (in standard lattice notation)
\begin{equation}
H = -t \sum_{\langle i,j \rangle} b_i^{\dagger} b_j - \sum_i (\mu - \epsilon_i) n_i
\end{equation}
for hard-core bosons with tunneling amplitude $t$ and chemical potential $\mu$. 
The $\epsilon_i$ are iid distributed according to a uniform distribution $[-\Delta, \Delta]$. The unit is the hopping $t=1$.

For densities $n=0$ and $n=1$ the system is a band insulator (empty or full). In the absence of disorder, the system is superfluid for any density at zero temperature.
Our goal is to determine the critical disorder $\Delta_c$ as a function of the defect density for which the system goes over from the superfluid into an insulating Bose glass phase.
The Bose glass phase has the counterintuitive properties that it is a compressible, gapless though gapless phase~\cite{Fisher89}.
Interest is especially towards low densities. The system has a "particle-hole" symmetry, so we only need to look at densities $n < 0.5$.

We will use path integral Monte Carlo simulations that use the worm algorithm~\cite{Prokofev98}, here in the implementation of Ref.~\onlinecite{Pollet07}.
The transition from the superfluid phase to the Bose glass phase is described by the vanishing of the superfluid density $\rho_s$ which is related to the winding number through the formula~\cite{Pollock87}
\begin{equation}
\rho_s = \frac{L^{2-d}\langle W^2 \rangle }{d \beta}.
\end{equation}
This reduces to $\rho_s = \langle W^2 \rangle / (2\beta)$ in 2d.  Note that we know from a general analytical argument
that the Bose glass always intervenes between the superfluid phase and the insulating vacuum or band insulator~\cite{Pollet09}.
The phase diagram of the above Hamiltonian has recently been mapped out for density $n=1$ in 3d for soft-core bosons~\cite{Gurarie09}. In 2d,
the phase diagram has not been published yet.

The quantum phase transition can be determined numerically from a proper finite size scaling analysis by letting the system size $L \to \infty$ and temperature $T \to 0$.
From Ref.~\onlinecite{Fisher89} we know that the dynamical critical exponent $z=d=2$, based on the finite value of the compressibility.
We can then study the quantum phase transition by taking different system sizes $L$ and scaling $\beta \sim L^2$.
We also know the following scaling behavior of the superfluid density in the vicinity of the critical point: 
\begin{equation}
\rho_s = \xi^{-1}f_s( {\xi / L} ) = L^{-1} \tilde{f}_s ( \delta L^{1/\nu}),
\end{equation}
because $\rho_s \sim \vert \delta \vert ^{\nu}$. Here, $\vert \delta \vert$ denotes the dimensionless detuning from the critical point, $\xi$ denotes the correlation length, $f_s$ is a universal scaling function,  and $\nu$ is the critical exponent. The second equality follows from the fact that the correlation length is cut off by the system size if it exceeds the system size.
If we scale $\beta \sim L^2$ then the curves $\rho_s L^2$ for different system sizes should intersect in a single point, provided we are in the scaling regime and that the irrelevant terms are sufficiently weak.
It is computationally advantageous to scale $\beta \sim L$. Then, the curves $\langle W^2 \rangle \sim \rho_s \beta \sim \rho_s L$ will not intersect in a single point. Next, we determine the intersection points between curves for consecutive system sizes, and extrapolate those intersection points to infinity. Both methods should of course lead to the same critical point. 

\begin{figure}
\centerline{\includegraphics[angle=-90, width=0.5\columnwidth]{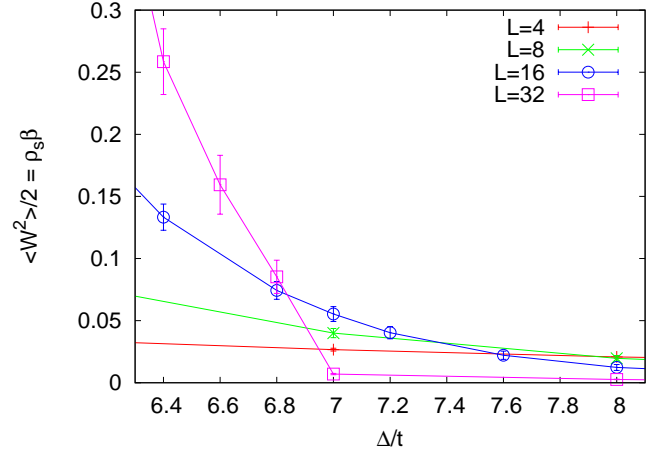}}
\caption{ (Color online). Curve for the square of winding number as a function of the disorder bound for $n=0.125$. Inverse temperature is scaled as $\beta = L^2 / 16$.  }
\label{fig:N0.125}
\end{figure}

\begin{figure}
\centerline{\includegraphics[angle=-90, width=0.5\columnwidth]{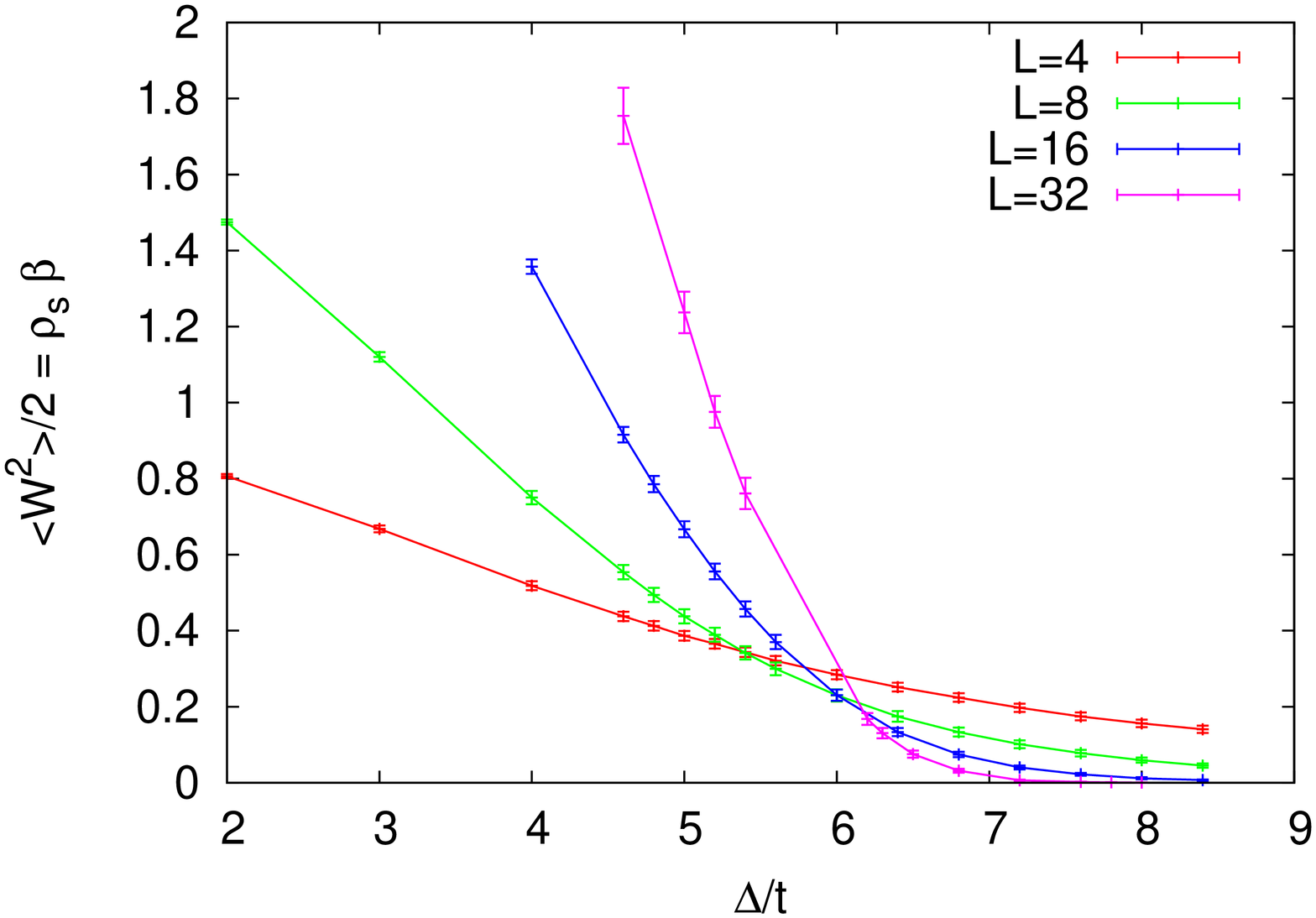}}
\caption{ (Color online). Curve for the square of the winding number as a function of the disorder bound for $n=0.125$. Inverse temperature is scaled as $\beta = L$.  }
\label{fig:N0.125_z1}
\end{figure}

\begin{figure}
\centerline{\includegraphics[angle=-90, width=0.5\columnwidth]{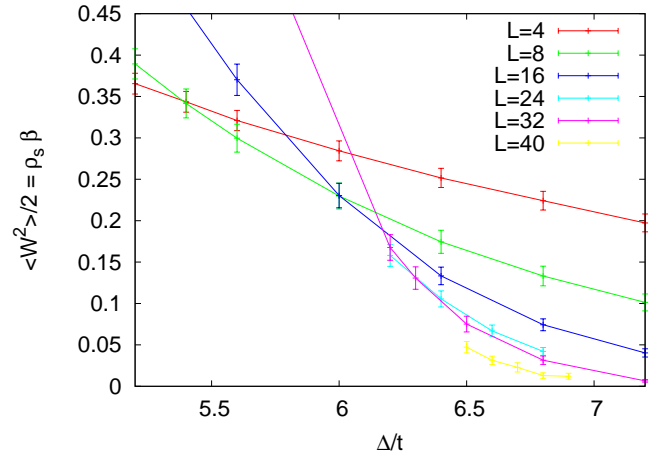}}
\caption{ (Color online). Same as in Fig.~\ref{fig:N0.125_z1}, but zooming in on the relevant intersection points.  }
\label{fig:N0.125_z1_bis}
\end{figure}

\begin{figure}
\centerline{\includegraphics[angle=-90, width=0.5\columnwidth]{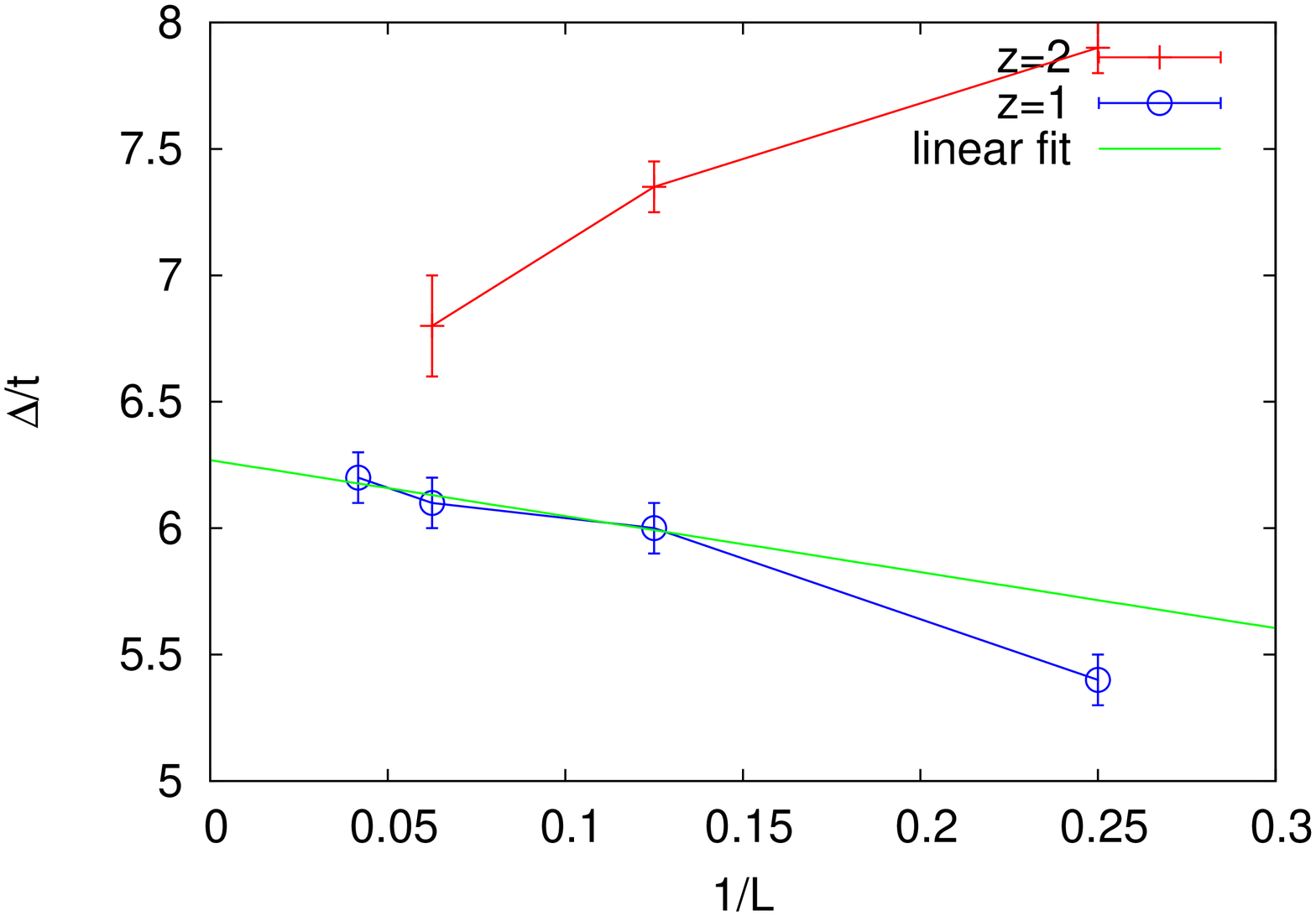}}
\caption{ (Color online).  Intersection points for $n=0.125$.  From these curves, the thermodynamic critical point is found to be $\Delta_c/t = 6.3(2)$. The data for $z=1$ seem compatible with linear extrapolation, except for the smallest system size which is not in the scaling regime. The data for $z=2$ are compatible with this extrapolation.}
\label{fig:cross_N0.125}
\end{figure}

We have studied the quantum phase transition at fixed density as follows. For every disorder realization we determined the chemical potential such that the average number of particles corresponds to our target density. We then computed the superfluid density as a function of the bound $\Delta$ for different system sizes $L=4,8, 16, 24, 32$. We typically averaged over 100 disorder realizations for the largest system sizes. In the first set of simulations, we scaled the inverse temperature as $\beta = L^2/16$, in the second set of simulations we scaled $\beta = L$, such that the simulations are identical for $L=16$. The second set of simulations was needed because we did not find a nice single intersection point in the first set of simulations in general, hinting that our temperatures are too high or that there may be strong finite size effects, which were difficult to overcome with $\beta \sim L^2$. In these simulations it is computationally better to use the scaling $\beta = L$, as well as to have a check on the data and the extrapolations.  We then look for the intersection points between consecutive curves. Representative data are shown in Fig.~\ref{fig:N0.125} for $n=0.125$ for quadratic scaling $(z=2)$ and in Fig.~\ref{fig:N0.125_z1} for linear scaling $(z=1)$. The intersection points are then shown in Fig.~\ref{fig:cross_N0.125} from which the quantum phase transition point can be determined quite accurately as $\Delta_c/t = 6.3(2)$. By repeating this procedure for the densities $n=1/100, n=1/16, n=1/4, n=3/8$ and $n=0.5$ we arrive at the phase diagram shown in Fig.~\ref{fig:phasediagram}.

\begin{figure}
\centerline{\includegraphics[angle=-90, width=0.5\columnwidth]{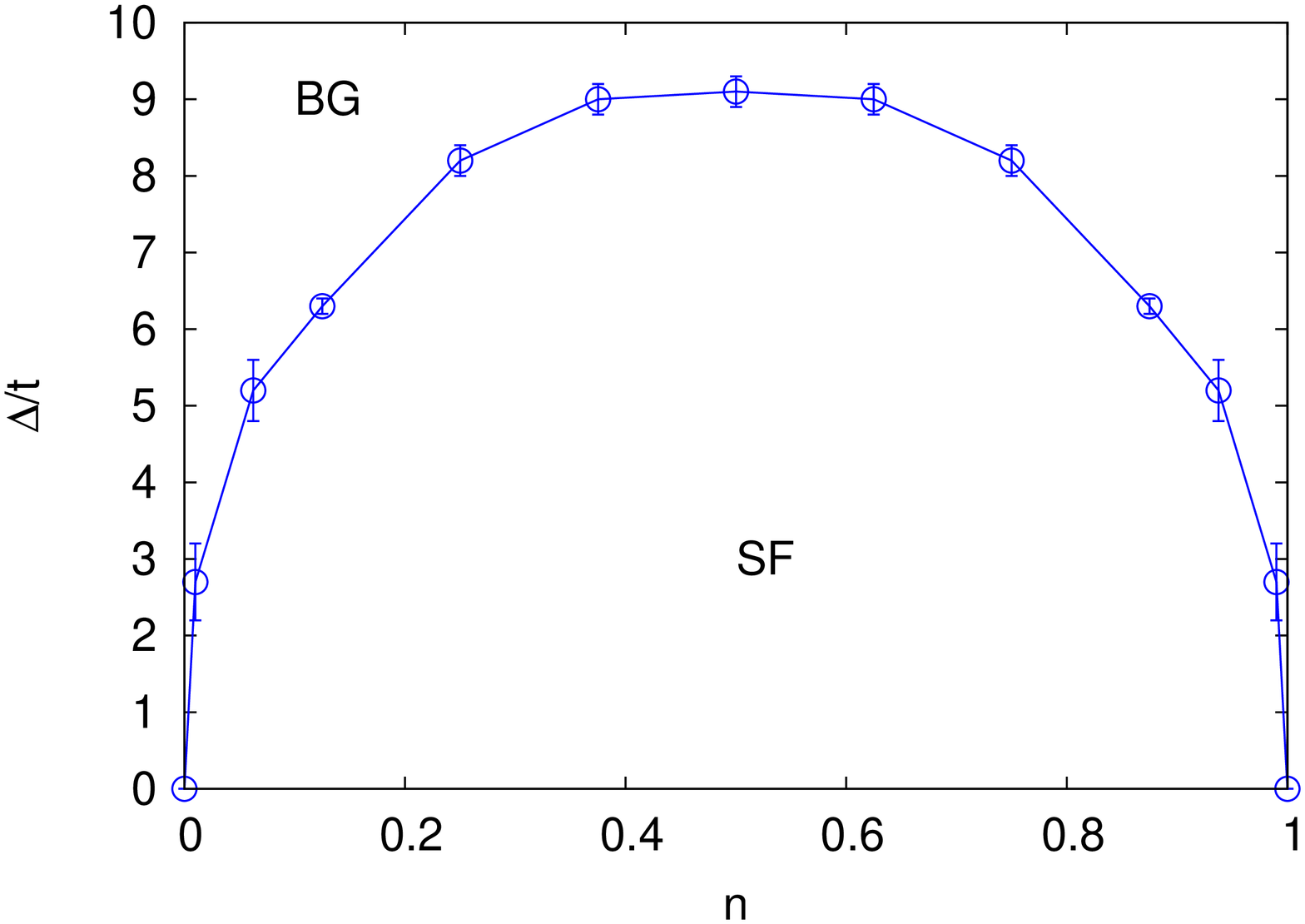}}
\caption{ (Color online). Shown is the transition from the superfluid (SF) to the Bose glass (BG) as a function of the density $n$ for a system of hard-core bosons on a square lattice and no further interactions between them at $T=0$ in the thermodynamic limit. }
\label{fig:phasediagram}
\end{figure}

\end{widetext}



\end{document}